%% file: main.tex
\documentclass{article}
\usepackage{fullpage}           
\usepackage{amsfonts}           
\usepackage{amsthm}             
\usepackage{amsmath}
\usepackage{amssymb}
\usepackage{xspace}             
\usepackage{graphicx}            
\usepackage{adjustbox}          

\usepackage{multicol}            
\usepackage{url}        
\usepackage{thmtools}
\usepackage{thm-restate}
\usepackage{enumerate}  
\usepackage{subcaption}
\usepackage[usenames]{xcolor} 

\usepackage{hyperref}
\usepackage{mathabx}

\usepackage{algorithm}
\usepackage[noend]{algorithmic}

\usepackage{nccmath}
\newtheorem{theorem}{Theorem}
\newtheorem{lemma}[theorem]{Lemma}

\newtheorem{remark}[theorem]{Remark}
\newtheorem{claim}[theorem]{Claim}

\newcommand{\fr}{\#}
\newcommand{\Oh}{O}
\newcommand{\tOh}{\widetilde{O}}
\newcommand{\tOmega}{\widetilde{\Omega}}
\newcommand{\polylog}{\textup{polylog}}
\newcommand{\Ex}{\mathbb{E}}
\newcommand{\Var}{\mathbb{V}}
\newcommand{\Ber}{\textup{Ber}}
\newcommand{\Bin}{\textup{Bin}}

\newcounter{sideremark}

\usepackage{setspace}

\title{A Linear-Time $n^{0.4}$-Approximation for Longest Common Subsequence} 
\date{}

\author{
  Karl Bringmann\thanks{Saarland University and Max-Planck-Institute for Informatics, Saarland Informatics Campus, Saarbr\"ucken, Germany.
   \texttt{bringmann@cs.uni-saarland.de}. This work is part of the project TIPEA that has received funding from the European Research Council (ERC) under the European Unions Horizon 2020 research and innovation programme (grant agreement No. 850979).
  }
  \and
  Vincent Cohen-Addad\thanks{
   Sorbonne Universit\'e, UPMC Univ Paris 06, CNRS, LIP6, Paris, France.
   \texttt{vcohenad@gmail.com}.
  }
  \and 
  Debarati Das\thanks{Basic Algorithm Research Copenhagen (BARC), University of Copenhagen, Denmark.
    \texttt{debaratix710@gmail.com}. Work supported by Basic Algorithms Research Copenhagen, grant 16582 from the VILLUM Foundation.
   } 
}

\begin{document}

\maketitle

\begin{abstract}
We consider the classic problem of computing the Longest Common Subsequence (LCS) of
two strings of length $n$. While a simple quadratic algorithm has been known for the problem for more than 40 years,
no faster algorithm has been found despite an extensive effort.
The lack of progress on the problem has recently been explained by Abboud, Backurs, and Vassilevska Williams [FOCS'15]
and Bringmann and Künnemann [FOCS'15] who proved that there is no subquadratic algorithm unless the Strong Exponential
Time Hypothesis fails. This major roadblock for getting faster exact algorithms has led the community to look for
subquadratic \emph{approximation} algorithms for the problem.

Yet, unlike the edit distance problem for which a constant-factor approximation in almost-linear time is known,
very little progress has been made on LCS, making it a notoriously difficult problem also in the realm of approximation.
For the general setting (where we make no assumption on the length of the optimum
solution or the alphabet size), only a naive $O(n^{\varepsilon/2})$-approximation algorithm with running time
$\tOh(n^{2-\varepsilon})$ has been known, for any constant $0 < \varepsilon \le 1$.
Recently, a breakthrough result by Hajiaghayi, Seddighin, Seddighin, and Sun [SODA'19] provided a linear-time
algorithm that yields a $O(n^{0.497956})$-approximation in expectation; improving upon the naive $O(\sqrt{n})$-approximation
for the first time.

In this paper, we provide an algorithm that in time $O(n^{2-\varepsilon})$ computes an $\tOh(n^{2\varepsilon/5})$-approximation with high probability, for any $0 < \varepsilon \le 1$. Our result (1) gives an
$\tOh(n^{0.4})$-approximation in linear time, improving upon the bound of
Hajiaghayi, Seddighin, Seddighin, and Sun, (2) provides an algorithm whose approximation scales
with any subquadratic running time $O(n^{2-\varepsilon})$, improving upon the naive bound of $O(n^{\varepsilon/2})$ for any $\varepsilon$,
and (3) instead of only in expectation, succeeds with high probability.

\end{abstract}

\section{Introduction}

The longest common subsequence (LCS) of two strings $x$ and $y$ is the longest string that appears as a subsequence of both strings.
The length of the LCS of $x$ and $y$, which we denote by $L(x,y)$, is one of the most fundamental measures of similarity between two strings and has drawn significant interest in last five decades, see, e.g.~\cite{WagnerF74, AhoHU76, Hirschberg77, HuntS77, MasekP80, NakatsuKY82, Apostolico86, Myers86, ApostolicoG87, WuMMM90, EppsteinGGI92, BergrothHR00, IliopoulosR09, AbboudBW15, BringmannK15, AbboudHWW16, BringmannK18, AbboudB18, RubinsteinSSS19, RubinsteinS20, HajiaghayiSSS19}. 
On strings of length $n$, the LCS problem can be solved exactly in quadratic time $O(n^2)$ using a classical dynamic programming approach~\cite{WagnerF74}. Despite an extensive line of research the quadratic running time has been improved only by logarithmic factors~\cite{MasekP80}. This lack of progress is explained by a recent result showing that any truly subquadratic algorithm for LCS would falsify the Strong Exponential Time Hypothesis (SETH); this has been proven independently by Abboud et al.~\cite{AbboudBW15} and by Bringmann and Künnemann~\cite{BringmannK15}. Further work in this direction shows that even a high polylogarithmic speedup for LCS would have surprising consequences~\cite{AbboudHWW16, AbboudB18}.
For the closely related edit distance the situation is similar, as the classic quadratic running time can be improved by logarithmic factors, but any truly subquadratic algorithm would falsify SETH~\cite{BackursI15}.

These strong hardness results naturally bring up the question whether LCS or edit distance can be efficiently \emph{approximated} (namely, whether
an algorithm with truly subquadratic time $O(n^{2-\varepsilon})$ for any constant $\varepsilon>0$, can produce a \emph{good} approximation in the worst-case).
In the last two decades, significant progress has been made towards designing efficient approximation algorithms for edit distance~\cite{BatuEKMRRS03, BarYossefJKK04, BatuES06, AndoniO12, AndoniKO10, ChakrabortyDGKS18, GoldenbergKS19, KouckyS20, BrakensiekR20}; the latest achievement is a constant-factor approximation in almost-linear\footnote{By \emph{almost-linear} we mean time $O(n^{1+\varepsilon})$ for a constant $\varepsilon>0$ that can be chosen arbitrarily small.} 
time~\cite{AndoniN20}.

For LCS the picture is much more frustrating. The LCS problem has a simple
$\tOh(n^{\varepsilon/2})$-approximation algorithm with running time 
$O(n^{2-\varepsilon})$ for any constant $0 < \varepsilon < 1$, and it has a trivial $|\Sigma|$-approximation algorithm with running time $O(n)$ for strings over alphabet $\Sigma$. Yet, improving upon these naive bounds has evaded the community until very recently, making
LCS a notoriously hard problem to approximate.
In 2019, Rubinstein et al.~\cite{RubinsteinSSS19} presented a subquadratic-time $O(\lambda^3)$-approximation, where $\lambda$ is the ratio of
the string length to the length of the optimal LCS. For binary alphabet, Rubinstein and Song~\cite{RubinsteinS20} recently improved the $2$\nobreakdash-approximation. 
In the general case (where $\lambda$ and the alphabet size are arbitrary), the naive $O(\sqrt{n})$-approximation in near-linear\footnote{By \emph{near-linear} we mean time $\tOh(n)$, where $\tOh$ hides polylogarithmic factors in $n$.} time was recently beaten by Hajiaghayi et al.~\cite{HajiaghayiSSS19}, who designed a linear-time algorithm
that computes an $O(n^{0.497956})$-approximation \emph{in expectation}.\footnote{\label{FN:highprob}While the SODA proceedings version
	of~\cite{HajiaghayiSSS19} claimed a high probability bound, the newer corrected Arxiv version~\cite{HajiaghayiSSS19arxiv} only claims
	that the algorithm outputs an $O(n^{0.497956})$-approximation in expectation. Personal communications with the authors confirm that
	the result indeed holds only in expectation,
	see also Remark~\ref{rem:alg2}.}
Nonetheless, the gap between the upper bound provided by
Hajiaghayi et al.~\cite{HajiaghayiSSS19} and the recent results on hardness of approximation~\cite{AbboudB17,AbboudR18} remains huge.

\subsection{Our Contribution}

We present a randomized $\tOh(n^{0.4})$-approximation for LCS running in linear time $O(n)$, where the approximation guarantee holds \emph{with high probability}\footnote{We say that an event happens \emph{with high probability} (w.h.p.) if it has probability at least $1-n^{-c}$, where the constant $c>0$ can be chosen in advance.}.
More generally, we obtain a tradeoff between approximation guarantee and running time: For any $0 < \varepsilon \le 1$ we achieve approximation ratio $\tOh(n^{2\epsilon/5})$ in time $O(n^{2-\varepsilon})$. Formally we prove the following: 

\begin{restatable}{theorem}{thmm}
	\label{thm:main}
	There is a randomized algorithm that, given strings $x,y$ of length $n \ge 1$ and a time budget $T \in [n,n^2]$, with high probability computes a multiplicative $\tOh(n^{0.8}/T^{0.4})$-approximation of the length of the LCS of $x$ and $y$ in time $O(T)$.
\end{restatable}

The improvement over the state of the art can be summarized as follows:
\begin{enumerate}
	\item An improved approximation ratio for the linear time regime: from $O(n^{0.497956})$~\cite{HajiaghayiSSS19}
	to $\tOh(n^{0.4})$;
	\item The first algorithm which improves upon the naive bound \emph{with high probability}$^4$;
	\item A generalization to running time $O(n^{2-\varepsilon})$, breaking the naive approximation ratio $\tOh(n^{\varepsilon/2})$ in general.
\end{enumerate}

\section{Technical Overview}
\label{sec:techniques}

We combine classic exact algorithms for LCS with different subsampling strategies to develop several algorithms that work in different regimes of the problem. A combination of these algorithms then yields the full approximation algorithm.

Our Algorithm 1 covers the regime of short LCS, i.e., when the LCS has length at most $n^\gamma$ for an appropriate constant $\gamma < 1$ depending on the running time budget. In this regime, we decrease the length of the string~$x$ by subsampling. This naturally allows to run classic exact algorithms for LCS on the subsampled string $x$ (which now has significantly smaller size) and the original string $y$, while not deteriorating the LCS between the two strings too much. 

For the remaining parts of the algorithm, the strings $x$ and $y$ are split into substrings $x_1,\ldots,x_{n/m}$ and $y_1,\ldots,y_{n/m}$ of length $m = n/\sqrt{T}$ where $T$ denotes the total running time budget. 
For any block~$(i,j)$ we write $L_{ij}$ for the length of the LCS of $x_i$ and $y_j$. 
We call a set ${\cal S} = \{ (i_1,j_1),\ldots,(i_k,j_k) \}$ with $i_1<\ldots<i_k$ and $j_1<\ldots<j_k$ a \emph{block sequence}. 
Since we can assume the LCS of $x$ and $y$ to be long, it follows that there exists a good ``block-aligned LCS'', more precisely there exists a block sequence with large LCS sum $\sum_{(i,j) \in {\cal S}} L_{ij}$. 

Now, a natural approach is to compute estimates $0 \le \widetilde L_{ij} \le L_{ij}$ for all blocks $(i,j)$ and to determine the maximum sum $\widetilde L = \sum_{(i,j) \in {\cal S}} \widetilde L_{ij}$ over all block sequences ${\cal S}$. Once we have estimates $\widetilde L_{ij}$, the maximum~$\widetilde L$ can be computed by dynamic programming in time $\Oh((n/m)^2)$, which is $\Oh(T)$ for our choice of $m$. 
In the following we describe three different strategies to compute estimates $\widetilde L_{ij}$. The major difficulty is that on average per block $(i,j)$ we can only afford time $\tOh(1)$ to compute an estimate $\widetilde L_{ij}$. 

The first strategy focuses on \emph{matching pairs}.
A matching pair of strings $s,t$ is a pair of indices $(a,b)$ such that $s[a] = t[b]$. We write $M_{ij}$ for the number of matching pairs of the strings $x_i$ and $y_j$.
Our Algorithm 2 
works well if some block sequence ${\cal S}$ has a large total number of matching pairs $\mu = \sum_{(i,j) \in {\cal S}} M_{ij}$. 
Here the key observation (Lemma~\ref{lem:apxsinglesymbol}) is that for each block $(i,j)$ there exists a symbol that occurs at least $\frac{M_{ij}}{2m}$ times in both $x_i$ and $y_j$. If $M_{ij}$ is large, matching this symbol provides a good approximation for $L_{ij}$. Unfortunately, since we can afford only $\tOh(1)$ running time per block, finding a frequent symbol is difficult. 
We develop as a new tool an algorithm that w.h.p.\ finds a frequent symbol in each block with an above-average number of matching pairs, see Lemma~\ref{lem:groupapxsinglesymbol}.

For our remaining two strategies we can assume the optimal LCS $L$ to be large and $\mu$ to be small (i.e., every block sequence has a small total number of matching pairs).
In our Algorithm 3, we analyze the case where $\lambda = \sum_{i,j} L_{ij}$ is large. 
Here we pick some diagonal and run our basic approximation algorithm on each block along the diagonal. 
Since there are $O(n/m)$ diagonals, an above-average diagonal has a total LCS of $\Omega(\lambda/(n/m))$. If~$\lambda$ is large then this provides a good estimation of the LCS. 
The main difficulty is how to find an above-average diagonal. A random diagonal has a good LCS sum in expectation, but not necessarily with good probability. Our solution is a non-uniform sampling, where we first test random blocks until we find a block with large LCS, and then choose the diagonal containing this seed block. 
This sampling yields an above-average diagonal with good probability.

Recall that there always exists a block sequence ${\cal G}$ with large LCS sum (see Lemma~\ref{lem:blockaligned}). The idea of our Algorithm 4 is to focus on a uniformly random subset of all blocks, where each block is picked with probability $p$. Then on each picked block we can spend more time (specifically time $\tOh(1/p)$) to compute an estimate $\widetilde L_{ij}$. Moreover, we still find a $p$-fraction of ${\cal G}$.
We analyze this algorithm in terms of $\mu$ and $\lambda$ (the choice of $p$ depends on these two parameters) and show that it works well in the complementary regimes of Algorithms 1-3.

\subparagraph*{Comparison with the Previous Approach of Hajiaghayi et al.~\cite{HajiaghayiSSS19}}
The general approach of splitting $x$ and $y$ into blocks and performing dynamic programming over estimates $\widetilde L_{ij}$ was introduced by Hajiaghayi et al.~\cite{HajiaghayiSSS19}. Moreover, our Algorithm 1 has essentially the same guarantees as \cite[Algorithm 1]{HajiaghayiSSS19}, but ours is a simple combination of generic parts that we reuse in our later algorithms, thus simplifying the overall algorithm.

Our Algorithm 2 follows the same idea as \cite[Algorithm 3]{HajiaghayiSSS19}, in that we want to find a frequent symbol in $x_i$ and $y_j$ and match only this symbol to obtain an estimate $\widetilde L_{ij}$. 
Hajiaghayi et al.\ find a frequent symbol by picking a \emph{random} symbol $\sigma$ in each block $x_i,y_j$; in expectation $\sigma$ appears at least $\frac{M_{ij}}{2m}$ times in $x_i$ and~$y_j$.
In order to obtain with high probability guarantees, we need to develop a new tool for finding frequent symbols not only in expectation but even with high probability, see Lemma~\ref{lem:groupapxsinglesymbol} and Remark~\ref{rem:alg2}.

The remainder of the approach differs significantly; our Algorithms 3 and 4 are very different compared to \cite[Algorithms 2 and 4]{HajiaghayiSSS19}. In the following we discuss their ideas.
In \cite[Algorithm 2]{HajiaghayiSSS19}, they argue about the alphabet size, splitting the alphabet into frequent and infrequent letters. For infrequent letters the total number of matching pairs is small, so augmenting a classic exact algorithm by subsampling works well. Therefore, they can assume that every letter is frequent and thus the alphabet size is small. We avoid this line of reasoning.
Finally, \cite[Algorithm 4]{HajiaghayiSSS19} is their most involved algorithm. Assuming that their other algorithms have failed to produce a sufficiently good approximation, they show that each part $x_i$ and $y_j$ can be turned into a \emph{semi-permutation} by a little subsampling. 
Then by leveraging Dilworth's theorem and Tu{\'r}an's theorem they show that most blocks have an LCS length of at least $n^{1/6}$; this
can be seen as a \emph{triangle inequality} for LCS and is their most novel contribution.
This results in a highly non-trivial algorithm making clever use of combinatorial machinery.

We show that these ideas can be completely avoided, by instead relying on classic algorithms based on matching
pairs augmented by subsampling.
Specifically, we replace their combinatorial machinery by our Algorithms 3 and 4 described above (recall that Algorithm 3 considers a non-uniformly sampled random diagonal while Algorithm 4 subsamples the set of blocks to be able to spend more time per block). 
We stress that our solution completely avoids the concept of semi-permutation or any heavy combinatorial machinery as used in~\cite[Algorithm 4]{HajiaghayiSSS19}, while providing a significantly improved approximation guarantee.

\subparagraph*{Organization of the Paper.}

Section~\ref{sec:prelim} introduces notation and a classical algorithm by Hunt and Szymanski.
In Section~\ref{sec:prep} we present our new tools, in particular for finding frequent symbols.
Section~\ref{sec:algo} contains our main algorithm, split into four parts that are presented in Sections~\ref{sec:algoone}, \ref{sec:algotwo}, \ref{sec:algothree}, and \ref{sec:algofour}, and combined in Section~\ref{sec:combine}.
In the appendix, for completeness we sketch the algorithm by Hunt and Szymanski (Appendix~\ref{app:prelim}) and we present pseudocode for all our algorithms (Appendix~\ref{app:pseudocodes}).

\section{Preliminaries}
\label{sec:prelim}

For $n \in \mathbb{N}$ we write $[n] = \{1,2,\ldots,n\}$. 
By the notation $\tOh$ and $\tOmega$ we hide factors of the form $\polylog(n)$. We use ``with high probability'' (w.h.p.) to denote probabilities of the form $1 - n^{-c}$, where the constant $c > 0$ can be chosen in advance.

\subparagraph*{String Notation.} A string $x$ over alphabet $\Sigma$ is a finite sequence of letters in $\Sigma$. We denote its length by $|x|$ and its $i$-th letter by $x[i]$. We also denote by $x[i..j]$ the substring consisting of letters $x[i] \ldots x[j]$. 
For any indices $i_1 < i_2 < \ldots < i_k$ the string $z = x[i_1]\ldots x[i_k]$ forms a \emph{subsequence} of $x$. 
For strings $x,y$ we denote by $L(x,y)$ the length of the longest common subsequence of $x$ and $y$. In this paper we study the problem of approximating $L(x,y)$ for given strings $x,y$ of length $n$. We focus on the length $L(x,y)$, however, our algorithms can be easily adapted to also reconstruct a subsequence attaining the output length. If $x,y$ are clear from the context, we may replace $L(x,y)$ by $L$.
Throughout the paper we assume that the alphabet is $\Sigma \subseteq [O(n)]$ (this is without loss of generality after a $\tOh(n)$-time preprocessing). 

\subparagraph*{Matching Pairs.}
For a symbol $\sigma \in \Sigma$, we denote the number of times that $\sigma$ appears in $x$ by $\fr_\sigma(x)$, and call this the \emph{frequency} of $\sigma$ in $x$.
For strings $x$ and $y$, a \emph{matching pair} is a pair $(i,j)$ with $x[i] = y[j]$. We denote the number of matching pairs by $M(x,y)$. If $x,y$ are clear from the context, we may replace $M(x,y)$ by $M$. Observe that
$M = \sum_{\sigma \in \Sigma}\fr_\sigma(x) \cdot  \fr_\sigma(y)$.
Using this equation we can compute $M$ in time $O(n)$. 

Hunt and Szymanski~\cite{HuntS77} solved the LCS problem in time $\tOh(n + M)$. More precisely, their algorithm can be viewed as having a preprocessing phase that only reads~$y$ and runs in time $\tOh(|y|)$, and a query phase that reads $x$ and $y$ and takes time $\tOh(|x| + M)$. 


\begin{restatable}[Hunt and Szymanski~\cite{HuntS77}]{theorem}{thmhs} \label{thm:hs77}
	We can preprocess a string $y$ in time $\tOh(|y|)$. Given a string $x$ and a preprocessed string $y$, we can compute their LCS in time $\tOh(|x| + M)$.
\end{restatable}

For convenience, we provide a proof sketch of their theorem in Appendix~\ref{app:prelim}.



\section{New Basic Tools}
\label{sec:prep}

\subsection{Basic Approximation Algorithm}
\label{sec:basicapx}

Throughout this section we abbreviate $L = L(x,y)$ and $M=M(x,y)$.
We start with the basic approximation algorithm that is central to our approach; most of our later algorithms use this as a subroutine. This algorithm subsamples the string $x$ and then runs Hunt and Szymanski's algorithm (Theorem~\ref{thm:hs77}).

\begin{restatable}[Basic Approximation Algorithm]{lemma}{lemb}
	\label{lem:basicapx}
	Let $x,y \in \Sigma^n$.
	We can preprocess $y$ in time $\tOh(n)$. Given~$x$, the preprocessed string $y$, and $\beta \ge 1$, in expected time $\tOh((n + M)/\beta + 1)$ we can compute 
	a value $\widetilde L \le L$ that w.h.p.\ satisfies $\widetilde L > \frac{L}{\beta} - 1$.
\end{restatable}

\begin{proof}
	In the preprocessing phase, we run the preprocessing of Theorem~\ref{thm:hs77} on~$y$. 
	
	Fix a constant $c \ge 1$. If $\beta \ge 1/(8c \log n)$, then in the query phase we simply run Theorem~\ref{thm:hs77}, solving LCS exactly in time $\tOh(|x| + M) = \tOh((n+M)/\beta + 1)$. 
	
	Otherwise, denote by $x'$ a random subsequence of $x$, where each letter $x[i]$ is removed independently with probability $1-p$ (i.e., kept with probability $p$) for $p := 8 c \log (n) / \beta$. Note that $p \le 1$ by our assumption on $\beta$. 
	We can sample $x'$ in expected time $O(|x'|+1)$, since the difference from one unremoved letter to the next is geometrically distributed, and geometric random variates can be sampled in expected time $O(1)$, see, e.g.,~\cite{BringmannF13}.
	Note that this subsampling yields $\Ex[|x'|] = p |x| = \tOh(|x| / \beta)$ and $\Ex[M(x',y)] = p\, M = \tOh(M / \beta)$.
	
	In the query phase, we sample $x'$ and then run the query phase of Theorem~\ref{thm:hs77} on $x'$ and~$y$. This runs in time $\tOh(|x'| + M(x',y) + 1)$, which is $\tOh((|x| + M)/\beta + 1)$ in expectation. 
	
	Finally, consider a fixed LCS of $x$ and $y$, namely $z = x[i_1]\ldots x[i_L] = y[j_1] \ldots y[j_L]$ for some $i_1 < \ldots < i_L$ and $j_1 < \ldots < j_L$. Each letter $x[i_k]$ survives the subsampling to $x'$ with probability~$p$. Therefore, we can bound $L(x',y)$ from below by a binomial random variable $\Bin(L,p)$ (the correct terminology is that $L(x',y)$ statistically dominates $\Bin(L,p)$). Since $Z = \Bin(L,p)$ is a sum of independent $\{0,1\}$-variables, multiplicative Chernoff applies and yields $\Pr[Z < \Ex[Z]/2] \le \exp(-\Ex[Z]/8)$. 
	If $L \ge \beta$ then $\Ex[Z] = L\, p \ge 2L/\beta$ and $\Ex[Z] \ge 8 c \log n$, and thus
	$\Pr[L(x',y) \ge L/\beta] \ge 1 - n^{-c}$. 
	Otherwise, if $L < \beta$, then we can only bound $L(x',y) \ge 0$. In both cases, we have $L(x',y) > L/\beta - 1$ with high probability.
\end{proof}

The above lemma behaves poorly if $L \le \beta$, due to the ``$-1$'' in the approximation guarantee. We next show that this can be avoided, at the cost of increasing the running time by an additive $\tOh(n)$.

\begin{lemma}[Generalised Basic Approximation Algorithm] \label{lem:extendedbasicapx}
	Given $x,y \in \Sigma^n$ and $\beta \ge 1$, in expected time $\tOh(n + M/\beta)$ we can compute a value $\widetilde L \le L$ that w.h.p.\ satisfies $\widetilde L \ge L/\beta$.
\end{lemma}
\begin{proof}
	We run the basic approximation algorithm from Lemma~\ref{lem:basicapx}, which computes a value $\widetilde L \le L$. Additionally, we compute the number of matching pairs $M = M(x,y)$ in time $\tOh(n)$. If $M > 0$, then there exists a matching pair, which yields a common subsequence of length 1. Therefore, if $M > 0$ we set $\widetilde L := \max\{\widetilde L, 1\}$.
	
	In the proof of Lemma~\ref{lem:basicapx} we showed that if $L \ge \beta$ then w.h.p.\ we have $\widetilde L \ge L/\beta$. We now argue differently in the case $L < \beta$. If $L = 0$, then $\widetilde L \ge 0 = L/\beta$ and we are done. If $0 < L < \beta$, then there must exist at least one matching pair, so $M > 0$, so the second part of our algorithm yields $\widetilde L \ge 1 > L/\beta$. Hence, in all cases w.h.p.\ we have $\widetilde L \ge L/\beta$.
\end{proof}

We now turn towards the problem of deciding for given $x,y$ and $\ell$ whether $L(x,y) \ge \ell$. To this end, we repeatedly call the basic approximation algorithm with geometrically decreasing approximation ratio $\beta$. 
Note that with decreasing approximation ratio we get a better approximation guarantee at the cost of higher running time. The idea is that if the LCS $L = L(x,y)$ is much shorter than the threshold $\ell$, then already approximation ratio $\beta \approx \ell/L$ allows us to detect that $L < \ell$. This yields a running time bound depending on the gap $L/\ell$.

\begin{restatable}[Basic Decision Algorithm]{lemma}{lemd}
	\label{lem:basicdec}
	Let $x,y \in \Sigma^n$.
	We can preprocess $y$ in time $\tOh(n)$. Given~$x$, the preprocessed $y$, and a number $1 \le \ell \le n$, in expected time $\tOh((n + M) L / \ell + n/\ell)$ we can w.h.p.\ correctly decide whether $L \ge \ell$.
	Our algorithm has no false positives (and w.h.p.\ no false negatives).
\end{restatable}

\begin{proof}
	In the preprocessing phase, we run the preprocessing of Lemma~\ref{lem:basicapx}. 
	In the query phase, we repeatedly call the query phase of Lemma~\ref{lem:basicapx}, with geometrically decreasing values of $\beta$:
	\begin{enumerate}
		\item[1.] Preprocessing: Run the preprocessing of Lemma~\ref{lem:basicapx}. 
		\item[2.] For $\beta = n, n/2, n/4,\ldots,1$:
		\begin{enumerate}
			\item[3.] Run the query phase of Lemma~\ref{lem:basicapx} with parameter $\beta$ to obtain an estimate $\widetilde L$.
			\item[4.] If $\widetilde L \ge \ell$: return ``$L \ge \ell$''
			\item[5.] If $\widetilde L \le \ell / \beta - 1$: return ``$L < \ell$''
		\end{enumerate}
	\end{enumerate}
	
	Let us first argue correctness. Since Lemma~\ref{lem:basicapx} computes a common subsequence of $x,y$, we have $\widetilde L \le L$. Thus, if $\widetilde L \ge \ell$, we correctly infer $L \ge \ell$. 
	Moreover, w.h.p.\ $\widetilde L$ satisfies $\widetilde L > L/\beta - 1$. Therefore, if $\widetilde L \le \ell/\beta - 1$, we can infer $L < \ell$, and this decision is correct with high probability.
	Finally, in the last iteration (where $\beta = 1$), we have $\ell/\beta - 1 = \ell - 1$, and thus one of $\widetilde L \ge \ell$ or $\widetilde L \le \ell/\beta-1$ must hold, so the algorithm indeed returns a decision.
	
	The expected time of the query phase of Lemma~\ref{lem:basicapx} is $\tOh((n + M)/\beta + 1)$. Since $\beta$ decreases geometrically, the total expected time of our algorithm is dominated by the last call. 
	
	If $L \ge \ell$, the last call is at the latest for $\beta = 1$. This yields running time $\tOh(n + M) \le \tOh((n+M)L/\ell)$.
	
	If $L < \ell$, note that for any $\beta \le \frac \ell {L+1}$ we have $\widetilde L \le L \le \ell/\beta - 1$, and thus we return ``$L < \ell$''. Because we decrease $\beta$ by a factor 2 in each iteration, the last call satisfies $\beta \ge \frac \ell{2(L+1)}$. Hence, the expected running time is  $\tOh((n + M)(L+1)/\ell + 1)$. If $L \ge 1$ then this time bound simplifies to $\tOh((n + M)L/\ell + 1)$. If $L = 0$, then also $M=0$, and the time bound becomes $\tOh(n/\ell + 1)$. In both cases we can bound the expected running time by the claimed $\tOh((n + M) L / \ell + n/\ell)$, since $\ell \le n$.
\end{proof}

\subsection{Approximating the Number of Matching Pairs}
\label{sec:apxMij}

Recall that for given strings $x,y$ of length $n$ the number of matching pairs $M=M(x,y)$ can be computed in time $O(n)$, which is linear in the input size. 
However, later in the paper we will split $x$ into substrings $x_1,\ldots,x_{n/m}$ and $y$ into substrings $y_1,\ldots,y_{n/m}$, each of length~$m$, and we will need estimates of the numbers of matching pairs $M_{ij} = M(x_i,y_j)$. 
In this setting, the input size is still $n$ (the total length of all strings $x_i$ and $y_j$) and the output size is $(n/m)^2$ (all numbers $M_{ij}$), but we are not aware of any algorithm computing the numbers $M_{ij}$ in near-linear time in the input plus output size $\tOh(n + (n/m)^2)$.\footnote{In fact, one can show conditional lower bounds from Boolean matrix multiplication that rule out near-linear time for computing all $M_{ij}$'s unless the exponent of matrix multiplication is $\omega = 2$.}
Therefore, we devise an approximation algorithm for estimating the number of matching pairs.

\begin{restatable}{lemma}{lemm}
	\label{lem:countmatchingpairs}
	For $x_1,\ldots,x_{n/m},y_1,\ldots,y_{n/m} \in \Sigma^m$ write $M_{ij} = M(x_i,y_j)$ and $M = \sum_{i,j} M_{ij}$. Given the strings $x_1,\ldots,x_{n/m},y_1,\ldots,y_{n/m}$ and a number $q > 0$, we can compute values $\widetilde M_{ij}$ that w.h.p.\ satisfy $M_{ij}/8 - q \le \widetilde M_{ij} \le 4 M_{ij}$, in total expected time $\tOh(n + M/q)$.
\end{restatable}

This yields a near-linear-time constant-factor approximation of all \emph{above-average} $M_{ij}$: By setting $q := \Theta(\frac{M m^2}{n^2})$, in expected time $\tOh(n + (n/m)^2)$ we obtain a constant-factor approximation of all values $M_{ij}$ with $M_{ij} \gg q$.

\begin{proof}
	The algorithm works as follows.
	\begin{enumerate}
		\item[1.] \emph{Graph Construction:} Build a three-layered graph $G$ on vertex set $V(G) = L \cup U \cup R$, where $L$ has a node $i$ for every string $x_i$, $R$ has a node $j$ for every string $y_j$, and $U$ has a node $(\sigma,\ell,r)$ for any $\sigma \in \Sigma$ and $0 \le \ell,r \le \log m$. Put an edge from $i \in L$ to $(\sigma,\ell,r) \in U$ iff $\fr_\sigma(x_i) \in [2^{\ell}, 2^{\ell+1})$. Similarly, put an edge from $j \in R$ to $(\sigma,\ell,r) \in U$ iff $\fr_\sigma(y_j) \in [2^{r}, 2^{r+1})$. Note that all frequencies and thus all edges of this graphs can be computed in total time $\tOh(n)$.
		For $i \in L$ and $j \in R$, we denote by $U_{ij} \subseteq U$ their common neighbors. Note that any $(\sigma,\ell,r) \in U_{ij}$ represents all matching pairs of symbol $\sigma$ in $x_i$ and $y_j$, and the number of these matching pairs is $\fr_\sigma(x_i) \cdot \fr_\sigma(y_j) \in [2^{\ell+r}, 2^{\ell+r+2})$. 
		\item[2.] \emph{Subsampling:} We sample a subset $\widetilde U \subseteq U$ by removing each node $(\sigma,\ell,r) \in U$ independently with probability $1 - p_{\ell,r}$, 
		where $p_{\ell,r} := \min\{1, 2^{\ell+r+3} / q\}$.
		\item[3.] \emph{Determine Common Neighbors:} For each $(\sigma,\ell,r) \in \widetilde U$ enumerate all pairs of neighbors $i \in L$ and $j \in R$. For each such 2-path, add $(\sigma,\ell,r)$ to an initially empty set $\widetilde U_{ij}$. This step computes the sets $\widetilde U_{ij} := U_{ij} \cap \widetilde U$ in time proportional to their total size.
		\item[4.] \emph{Output:} Return the values $\widetilde M_{ij} := \sum_{(\sigma,\ell,r) \in \widetilde U_{ij}} 2^{\ell+r} / p_{\ell,r}$.
	\end{enumerate}
	
	\subparagraph*{Correctness:}
	To analyze this algorithm, we consider the numbers $\overline M_{ij} := \sum_{(\sigma,\ell,r)\in U_{ij}} 2^{\ell+r}$. Observe that we have $\overline M_{ij} \le M_{ij} \le 4 \overline M_{ij}$, since each $(\sigma,\ell,r)\in U_{ij}$ corresponds to at least $2^{\ell+r}$ and at most $2^{\ell+r+2}$ matching pairs of $x_i$ and $y_j$. It therefore suffices to show that $\widetilde M_{ij}$ is close to $\overline M_{ij}$.
	Using Bernoulli random variables $\Ber(p_{\ell,r})$ to express whether $(\sigma,\ell,r)$ survives the subsampling, we write 
	\[ \widetilde M_{ij} = \sum_{(\sigma,\ell,r) \in U_{ij}} \frac{2^{\ell+r}}{p_{\ell,r}} \cdot \Ber(p_{\ell,r}). \]
	This yields an expected value of $\Ex[\widetilde M_{ij}] = \overline M_{ij}$, so by Markov's inequality we obtain $\widetilde M_{ij} \le 4 \overline M_{ij} \le 4 M_{ij}$ with probability at least $3/4$. 
	Since $\widetilde M_{ij}$ is a linear combination of independent Bernoulli random variables, we can also easily express its variance as
	\[ \Var[\widetilde M_{ij}] = \sum_{(\sigma,\ell,r) \in U_{ij}} \big(2^{\ell+r} / p_{\ell,r} \big)^2 \cdot p_{\ell,r}(1 - p_{\ell,r}) = \sum_{(\sigma,\ell,r) \in U_{ij}} 2^{\ell+r} \cdot 2^{\ell+r} \Big(\frac1{p_{\ell,r}} - 1\Big). \]
	We now use the definition of $p_{\ell,r} := \min\{1, 2^{\ell+r+3} / q\}$ to bound
	\[ 2^{\ell+r} \Big(\frac1{p_{\ell,r}} - 1\Big) = 2^{\ell+r} \Big( \max\Big\{1,\frac{q}{2^{\ell+r+3}}\Big\} - 1\Big) = \max\{0, q/8 - 2^{\ell+r}\} \le q/8. \]
	This yields
	$ \Var[\widetilde M_{ij}]  \le \overline M_{ij} q/8$.
	We now use Chebychev's inequality $\Pr[X < \Ex[X] - \lambda] \le \Var[X] / \lambda^2$ on $\lambda = 0.5 \Ex[X]$ and $X = \widetilde M_{ij}$ to obtain
	\[ \Pr[ \widetilde M_{ij} < \overline M_{ij}/2] \le \frac{q}{2 \overline M_{ij}}. \]
	In case $M_{ij} \ge 8 q$, we have $\overline M_{i,j} \ge M_{ij}/4 \ge 2 q$ and hence $\Pr[\widetilde M_{ij} \ge M_{ij}/8] \ge \Pr[\widetilde M_{ij} \ge \overline M_{ij}/2] \ge 3/4$. 
	Otherwise, in case $M_{ij} < 8 q$, we can only use the trivial $\widetilde M_{ij} \ge 0 > M_{ij}/8 - q$. 
	
	Hence, each inequality $\widetilde M_{ij} \le 4 M_{ij}$ and $\widetilde M_{ij} \ge M_{ij}/8 - q$ individually holds with probability at leat $3/4$.
	Finally, we boost the success probability by repeating the above algorithm $\Oh(\log n)$ times and returning for each $i,j$ the median of all computed values $\widetilde M_{ij}$. 
	
	\subparagraph*{Running Time:}
	Steps 1 and 2 can be easily seen to run in time $\tOh(n)$. 
	Steps 3 and 4 run in time proportional to the total size of all sets $\widetilde U_{ij}$, which we claim to be at most $8M/q$ in expectation. Over $\Oh(\log n)$ repetitions, we obtain a total expected running time of $\tOh(n + M/q)$. (We remark that here we consider a succinct output format, where only the non-zero numbers $\widetilde M_{ij}$ are listed; otherwise additional time of $\tOh((n/m)^2)$ is required to output the numbers $\widetilde M_{ij} = 0$.)
	
	It remains to prove the claimed bound of $\Ex[\sum_{i,j} |\widetilde U_{ij}|] \le 8M/q$.
	Since $2^{\ell+r} / p_{\ell,r} = \max\{2^{\ell+r}, q/8\} \ge q/8$, from the definition of $\widetilde M_{ij} = \sum_{(\sigma,\ell,r) \in \widetilde U_{ij}} 2^{\ell+r} / p_{\ell,r}$ we infer $\widetilde M_{ij} \ge \frac q8 |\widetilde U_{ij}|$. Therefore,
	\[\Ex\Big[\sum_{i,j} |\widetilde U_{ij}|\Big] \le \Ex\Big[\frac 8q \sum_{i,j} \widetilde M_{ij}\Big] = \frac 8q \sum_{i,j} \overline M_{ij} \le \frac 8q \sum_{i,j} M_{ij} = \frac {8M}q. \qedhere \]
\end{proof}

\subsection{Single Symbol Approximation Algorithm}
\label{sec:prepsingleapx}

For strings $x,y$ that have a large number of matchings pairs $M = M(x,y)$, some symbol must appear often in $x$ and in $y$. 
This yields a common subsequence using (several repetitions of) a single alphabet symbol.

\begin{restatable}[Cf.\ Lemma 6.6.(ii) in~\cite{BringmannK18} or Algorithm 3 in~\cite{HajiaghayiSSS19}]{lemma}{lems}
	\label{lem:apxsinglesymbol}
	For any $x,y \in \Sigma^n$ 
	there exists a symbol $\sigma \in \Sigma$ that appears at least $\frac {M}{2n}$ times in $x$ and in $y$. Therefore, in time $\tOh(n)$ we can compute a common subsequence of $x,y$ of length at least $\frac {M}{2n}$. 
	In particular, we can compute a value $\widetilde L \le L$ that satisfies $\widetilde L \ge \frac{M}{2n}$. 
\end{restatable}

\begin{proof}
	Let $k$ be maximal such that some symbol $\sigma \in \Sigma$ appears at least $k$ times in $x$ and at least $k$ times in~$y$. Let $\Sigma^w := \{ \sigma \in \Sigma \mid \fr_\sigma(w) \le k\}$ for $w \in \{x,y\}$. Since no symbol appears more than $k$ times in $x$ and in $y$, we have $\Sigma^x \cup \Sigma^y = \Sigma$. We can thus bound
	\[ M = M(x,y) = \sum_{\sigma \in \Sigma} \fr_\sigma(x) \cdot \fr_\sigma(y) \le \sum_{\sigma \in \Sigma^x} k \cdot \fr_\sigma(y) + \sum_{\sigma \in \Sigma^y} \fr_\sigma(x) \cdot k \le 2 k n, \]
	since the frequencies $\fr_\sigma(x)$ sum up to at most $n$, and similarly for $\fr_\sigma(y)$. It follows that $k \ge \frac M{2n}$. 
	Computing $k$, and a symbol $\sigma \in \Sigma$ attaining $k$, in time $\tOh(n)$ is straightforward.
\end{proof}

We devise a variant of Lemma~\ref{lem:apxsinglesymbol} in the following setting. For strings $x_1,\ldots,x_{n/m}$, $y_1,\ldots,y_{n/m} \in \Sigma^m$ we write $L_{ij} = L(x_i,y_j)$, $M_{ij} = M(x_i,y_j)$ and $M = \sum_{i,j} M_{ij}$. We want to find for each block~$(i,j)$ a frequent symbol in $x_i$ and $y_j$, or equivalently we want to find a common subsequence of $x_i$ and $y_j$ using a single alphabet symbol. Similarly to Lemma~\ref{lem:countmatchingpairs}, we relax Lemma~\ref{lem:apxsinglesymbol} to obtain a fast running time.

\begin{restatable}{lemma}{lemsg}
	\label{lem:groupapxsinglesymbol}
	Given $x_1,\ldots,x_{n/m},y_1,\ldots,y_{n/m} \in \Sigma^m$ and any $q > 0$, we can compute for each $i,j$ a number~$\widetilde L_{ij} \le L_{ij}$ such that w.h.p.\ $\widetilde L_{ij} \ge \frac{M_{ij} - q}{16 m}$. The algorithm runs in total expected time $\tOh(n + M/q)$. 
\end{restatable}

\begin{proof}
	We run the same algorithm as in Lemma~\ref{lem:countmatchingpairs}, except that in Step 4 for each $i,j$ with non-empty set~$\widetilde U_{ij}$ we let $\widetilde L_{ij}$ be the maximum of $2^{\min\{\ell,r\}}$ over all $(\sigma,\ell,r) \in \widetilde U_{ij}$. For each empty set $\widetilde U_{ij}$, we implicitly set $\widetilde L_{ij} = 0$, i.e., we output a sparse representation of all non-zero values $\widetilde L_{ij}$.
	
	The running time analysis is the same as in Lemma~\ref{lem:countmatchingpairs}.
	
	For the upper bound on $\widetilde L_{ij}$, since $\sigma$ appears at least $2^\ell$ times in $x_i$ and at least $2^r$ times in $y_j$, there is a common subsequence of $x_i$ and $y_j$ of length at least $\widetilde L_{ij}$. Thus, we have $\widetilde L_{ij} \le L_{ij}$. 
	
	For the lower bound on $\widetilde L_{ij}$, fix $i,j$ and order the tuples $(\sigma,\ell,r) \in U_{ij}$ in ascending order of $2^{\min\{\ell,r\}}$, obtaining an ordering $(\sigma_1,\ell_1,r_1),\ldots,(\sigma_k,\ell_k,r_k)$. 
	For $h \in [k]$ we let ${\cal S} := \{(\sigma_1,\ell_1,r_1),\ldots,(\sigma_h,\ell_h,r_h)\}$ and ${\cal L} := \{(\sigma_h,\ell_h,r_h),\ldots,(\sigma_k,\ell_k,r_k)\}$. Recall that $\overline M_{ij} = \sum_{(\sigma,\ell,r) \in U_{ij}} 2^{\ell+r}$, and observe that we can pick $h$ with
	
	\begin{align} \label{eq:prepsingleone}
	\sum_{(\sigma,\ell,r) \in {\cal S}} 2^{\ell+r} \ge \overline M_{ij} / 2 \qquad \text{and} \qquad  \sum_{(\sigma,\ell,r) \in {\cal L}} 2^{\ell+r} \ge \overline M_{ij} / 2.
	\end{align}
	
	Then we have 
	\[ \frac{\overline M_{ij}}2 \le \sum_{(\sigma,\ell,r) \in {\cal S}} 2^{\ell+r} = \sum_{(\sigma,\ell,r) \in {\cal S}} 2^{\min\{\ell,r\}} \cdot 2^{\max\{\ell,r\}} \le 2^{\min\{\ell_h,r_h\}} \sum_{(\sigma,\ell,r) \in {\cal S}} 2^{\max\{\ell,r\}}. \]
	Note that for any $(\sigma,\ell,r) \in {\cal S}$ the symbol $\sigma$ appears at least $2^{\max\{\ell,r\}}$ times in $x_i$ or in $y_j$, and thus the sum on the right hand side is at most $2m$. Rearranging, this yields $2^{\min\{\ell_h,r_h\}} \ge \frac{\overline M_{ij}}{4m} \ge \frac{M_{ij}}{16m}$, where we used $\overline M_{ij} \ge M_{ij}/4$ as in the proof of Lemma~\ref{lem:countmatchingpairs}. In particular, due to our ordering we have for any $(\sigma,\ell,r) \in {\cal L}$:
	
	\begin{align} \label{eq:prepsingletwo}
	2^{\min\{\ell,r\}} \ge 2^{\min\{\ell_h,r_h\}} \ge \frac{M_{ij}}{16m}.
	\end{align}
	
	Consider the number of nodes in ${\cal L}$ surviving the subsampling, i.e., $Z := | {\cal L} \cap \widetilde U_{ij} |$.
	If $Z > 0$, then some node in ${\cal L}$ survived, and thus by (\ref{eq:prepsingletwo}) the computed value $\widetilde L_{ij}$ is at least $\frac{M_{ij}}{16 m}$. 
	It thus remains to analyze $\Pr[Z > 0]$.
	
	In case some $(\sigma,\ell,r) \in {\cal L}$ has $p_{\ell,r} = 1$, we have $Z > 0$ with probability 1. 
	Otherwise all $(\sigma,\ell,r) \in {\cal L}$ have $p_{\ell,r} < 1$ and thus $p_{\ell,r} = 2^{\ell+r+3}/q$. 
	In this case, we write $Z$ as a sum of independent Bernoulli random variates in the form $Z = \sum_{(\sigma,\ell,r) \in {\cal L}} \Ber(p_{\ell,r})$. In particular, 
	\[\Ex[Z] = \sum_{(\sigma,\ell,r) \in {\cal L}} 2^{\ell+r+3}/q \stackrel{(\ref{eq:prepsingleone})}{\ge} \frac {4 \overline M_{ij}}q \ge \frac{M_{ij}}q. \]
	Since $Z$ is a sum of independent $\{0,1\}$-variables, multiplicative Chernoff applies and yields $\Pr[Z < \Ex[Z]/2] \le \exp(-\Ex[Z]/8)$. We thus obtain 
	\[ \Pr[Z > 0] \ge 1 - \Pr\big[Z < \Ex[Z]/2\big] \ge 1 - \exp\big(-\Ex[Z]/8\big) \ge 1 - \exp\Big(-\frac{M_{ij}}{8q}\Big). \]
	In case $M_{ij} \ge q$, we obtain $\Pr[Z > 0] \ge 1 - \exp(-1/8) \ge 0.1$, and thus we have $\widetilde L_{ij} \ge \frac{M_{ij}}{16 m}$ with probability at least $0.1$. 
	Otherwise, in case $M_{ij} < q$, we can only use the trivial bound $\widetilde L_{ij} \ge 0 > \frac{M_{ij} - q}{16 m}$. 
	In any case, we have $\widetilde L_{ij} \ge \frac{M_{ij} - q}{16 m}$ with probability at least $0.1$. Similar to the proof of Lemma~\ref{lem:countmatchingpairs}, we run $\Oh(\log n)$ independent repetitions of this algorithm and return for each $i,j$ the maximum of all computed values $\widetilde L_{ij}$, to boost the success probability and finish the proof.
\end{proof}

\section{Main Algorithm}
\label{sec:algo}
In this section we prove Theorem~\ref{thm:main}. First we show that Theorem~\ref{thm:maintechnical} implies Theorem~\ref{thm:main}, and then in the remainder of this section we prove Theorem~\ref{thm:maintechnical}. 

\begin{restatable}[Main Result, Relaxation]{theorem}{thmg}
	\label{thm:maintechnical}
	Given strings $x,y$ of length $n$ and a time budget $T \in [n,n^2]$, in expected time $\tOh(T)$ we can compute a number $\widetilde L$ such that $\widetilde L \le L := L(x,y)$ and w.h.p.\ $\widetilde L \ge \tOmega(L T^{0.4} / n^{0.8})$.
\end{restatable}

Recall Theorem~\ref{thm:main}:

\thmm*

\begin{proof}[Proof of Theorem~\ref{thm:main} assuming Theorem~\ref{thm:maintechnical}]
	Note that the difference between Theorems~\ref{thm:main} and~\ref{thm:maintechnical} is that the latter allows \emph{expected} running time and has an additional slack of logarithmic factors in the running time.
	
	
	
	In order to remove the \emph{expected} running time, we abort the algorithm from Theorem~\ref{thm:maintechnical} after $\tOh(T)$ time steps. By Markov's inequality, we can choose the hidden constants and logfactors such that the probability of aborting is at most $1/2$. We boost the success probability of this adapted algorithm by running $\Oh(\log n)$ independent repetitions and returning the maximum over all computed values $\widetilde L$. This yields an $\tOh(n^{0.8}/T^{0.4})$-approximation with high probability in time $\tOh(T)$. 
	
	To remove the logfactors in the running time, as the first step in our algorithm we subsample the given strings $x,y$, keeping each symbol independently with probability $p = 1/\polylog(n)$, resulting in subsampled strings $\tilde x, \tilde y$. Since any common subsequence of $\tilde x, \tilde y$ is also a common subsequence of $x,y$, the estimate $\widetilde L$ that we compute for $\tilde x, \tilde y$ satisfies $\widetilde L \le L(\tilde x, \tilde y) \le L(x,y)$. Moreover, if $L(x,y) \ge \polylog(n)$ then by Chernoff bound with high probability we have $L(\tilde x, \tilde y) = \widetilde \Omega(L(x,y))$, so that an $\tOh(n^{0.8}/T^{0.4})$-approximation on $\tilde x,\tilde y$ also yields an $\tOh(n^{0.8}/T^{0.4})$-approximation on $x,y$. Otherwise, if $L(x,y) \le \polylog(n)$, then in order to compute a $\tOh(1)$-approximation it suffices to compute an LCS of length 1, which is just a matching pair and can be found in time $O(n)$ (assuming that the alphabet is $[O(n)]$).
	
	This yields an algorithm that computes a value $\widetilde L \le L$ such that w.h.p.\ $\widetilde L \ge \widetilde \Omega( L T^{0.4} / n^{0.8})$. The algorithm runs in time $O(T)$, and this running time bound holds deterministically, i.e., with probability 1. Hence, we proved Theorem~\ref{thm:main}.
\end{proof}

It remains to prove Theorem~\ref{thm:maintechnical}. Our algorithm is a combination of four methods that work well in different regimes of the problem, see Sections~\ref{sec:algoone}, \ref{sec:algotwo}, \ref{sec:algothree}, and \ref{sec:algofour}. We will combine these methods in Section~\ref{sec:combine}.

\subsection{Algorithm 1: Small \boldmath$L$}
\label{sec:algoone}

Algorithm 1 works well if the LCS is short. It yields the following result.


\begin{restatable}[Algorithm 1]{theorem}{thma}
	\label{thm:alg1}
	We can compute in expected time $\tOh(T)$ an estimate $\widetilde L \le L$ that w.h.p.\ satisfies $\widetilde L \ge \min\{ L, \sqrt{LT/n} \}$.
\end{restatable}

\begin{proof}
	Our Algorithm 1 works as follows.
	\begin{enumerate}
		\item[1.] Run Lemma~\ref{lem:apxsinglesymbol} on $x$ and $y$.
		\item[2.] Run Lemma~\ref{lem:extendedbasicapx} on $x$ and $y$ with $\beta := \max\{1, \frac M{2T}\}$.
		\item[3.] Output the larger of the two common subsequence lengths computed in Steps 1 and 2.
	\end{enumerate}
	
	\subparagraph*{Running Time:} Step 1 runs in time $\tOh(n) = \tOh(T)$. Step 2 runs in expected time $\tOh(n + M/\beta)$. Since $\beta \ge \frac M{2T}$ we have $M/\beta \le 2T$, so the expected running time is $\tOh(n+T) = \tOh(T)$. 
	
	\subparagraph*{Upper Bound:} Steps 1 and 2 compute common subsequences, so the computed estimate $\widetilde L$ satisfies $\widetilde L \le L$. 
	
	\subparagraph*{Approximation Guarantee:} 
	Note that Step 1 guarantees $\widetilde L \ge \frac M{2n}$ and Step 2 guarantees w.h.p.\ $\widetilde L \ge L/\beta$. 
	If $M \le 2T$ then $\beta = 1$ and $\widetilde L = L$, so we solved the problem exactly. Otherwise we have $M > 2T$ and $\beta = \frac M{2T}$, so Step 2 guarantees w.h.p.\ $\widetilde L \ge 2LT/M$. By multiplying the two guarantees on $\widetilde L$ and taking square roots, we obtain w.h.p.\
	\[ \widetilde L \ge \sqrt{ \frac M{2n} \cdot \frac {2LT}M } = \sqrt{ \frac{LT}n }. \]
	It follows that w.h.p.\ $\widetilde L \ge \min\{ L, \sqrt{LT/n} \}$. 
\end{proof}

\subsection{Block Sequences and Parameter Guessing}
\label{sec:blocklcs}

This section introduces some general notation and structure for the remaining algorithms.

\subparagraph*{Block Sequences:}
We split $x$ into substrings $x_1,\ldots,x_{n/m}$ of length $m = n/\sqrt{T}$. 
Similarly, we split $y$ into $y_1,\ldots,y_{n/m}$. 
A pair $(i,j) \in [n/m]^2$, corresponding to the substrings $x_i,y_j$, is called a \emph{block}.
For any block we write $M_{ij} = M(x_i,y_j)$ and $L_{ij} = L(x_i,y_j)$. 
Moreover, we write $(i,j) < (i',j')$ if and only if $i < i'$ and $j < j'$.
A \emph{block sequence} is a set ${\cal S} = \{(i_1,j_1),\ldots,(i_k,j_k) \}$ with ${\cal S} \subseteq [n/m]^2$ satisfying the monotonicity property $(i_1,j_1) < \ldots < (i_k,j_k)$. 
In what follows, every algorithm will compute estimates $0 \le \widetilde L_{ij} \le L_{ij}$ and then choose a block sequence~${\cal S}$ to produce an overall estimate $\widetilde L = \sum_{(i,j) \in {\cal S}} \widetilde L_{ij}$. Note that this guarantees $\widetilde L \le L$, as the sum $\sum_{(i,j) \in {\cal S}} \widetilde L_{ij}$ corresponds to some (block-aligned) common subsequence of $x$ and $y$.
In order to get bounds in the other direction, we need to show that there always exists a block sequence of large LCS sum, i.e., a long ``block-aligned common subsequence''. This is shown by the following lemma. 

\begin{lemma} \label{lem:blockaligned}
	There exists a block sequence ${\cal G}$ of size $|{\cal G}| = \frac {L \sqrt{T}}{8n}$ such that for any $(i,j) \in {\cal G}$ we have $L_{ij} \ge \frac{L}{4 \sqrt{T}}$ and $M_{ij} \le \frac{8 \mu n}{L \sqrt{T}}$. In particular, we have $\sum_{(i,j) \in {\cal G}} L_{ij} \ge \frac{L^2}{32 n}$.
\end{lemma}

\begin{remark}
	This is analogous to \cite[Lemma 8.2]{HajiaghayiSSS19}, but we improve the size of ${\cal G}$.
\end{remark}

\begin{proof}
	Let $L^*_{ij}$ be the \emph{contribution of block $(i,j)$} to the LCS.
	More precisely, fix an LCS $z$ of $x$ and $y$, and write $z = x[a_1]\ldots x[a_L] = y[b_1]\ldots y[b_L]$ for $(a_1,b_1) < \ldots < (a_L,b_L)$. 
	Then for any block $(i,j)$, the number $L^*_{ij}$ counts all indices $k$ with $a_k \in ((i-1)m,im]$ and $b_k \in ((j-1)m,jm]$. 
	Consider the set ${\cal A} := \{(i,j) \mid L^*_{ij} > 0\}$ consisting of all contributing blocks. 
	From the monotonicity $(a_1,b_1) < \ldots < (a_L,b_L)$ it follows that also the contributing blocks form a monotone sequence, in the sense that for any $(i,j),(i',j') \in {\cal A}$ we have $i \le i'$ and $j \le j'$, or $i' \le i$ and $j' \le j$. (However, these inequalities are not necessarily strict, so ${\cal A}$ is not necessarily a block sequence.)
	This monotonicity implies that there are $|{\cal A}| \le 2n/m$ contributing blocks.
	Also note that $\sum_{(i,j) \in {\cal A}} L^*_{ij} = L$. 
	Now consider the subset ${\cal B} = \{(i,j) \mid L^*_{ij} > \frac{L m}{4 n}\} \subseteq {\cal A}$. Note that the remaining blocks in total contribute 
	\[ \sum_{(i,j) \in {\cal A} \setminus {\cal B}} L^*_{ij} \le |{\cal A}| \cdot \frac{L m}{4 n} \le \frac{2n}m \cdot \frac{L m}{4 n} = \frac L 2, \]
	and thus ${\cal B}$ contributes $\sum_{(i,j) \in {\cal B}} L^*_{ij} \ge L/2$. 
	
	We now greedily pick a subset ${\cal C} \subseteq {\cal B}$ as follows. Pick any $(i,j) \in {\cal B}$, add $(i,j)$ to ${\cal C}$, and then remove each $(i',j') \in {\cal B}$ with $i'=i$ or $j'=j$ from ${\cal B}$. Repeat until ${\cal B}$ is empty. 
	
	By construction, ${\cal C}$ is a block sequence and for any $(i,j) \in {\cal C}$ we have $L_{ij} \ge \frac{L m}{4 n} = \frac{L}{4 \sqrt{T}}$.
	We claim that $|{\cal C}| \ge \frac{L}{4m} = \frac {L \sqrt{T}}{4n}$. 
	To see this, observe that all blocks $(i',j') \in {\cal B}$ with $i'=i$ in total contribute at most $m$, since they describe a subsequence of $x_i$, which has length $m$. Similarly, all blocks $(i',j') \in {\cal B}$ with $j'=j$ in total contribute at most $m$. Therefore, one step of the greedy procedure removes a contribution of at most $2m$. Since the total contribution is $\sum_{(i,j) \in {\cal B}} L^*_{ij} \ge L/2$, there are at least $\frac L{4m} = \frac {L \sqrt{T}}{4n}$ greedy steps. 
	Finally, we consider the number of matching pairs. Since ${\cal C}$ is a block sequence, we have $\sum_{(i,j) \in {\cal C}} M_{ij} \le \mu$. Thus, on average each $(i,j) \in {\cal C}$ has a number of matching pairs of at most $\mu / |{\cal C}| = \frac{4 \mu n}{L \sqrt{T}}$. 
	By Markov's inequality, at least half of the blocks $(i,j) \in {\cal C}$ have $M_{ij} \le \frac{8 \mu n}{L \sqrt{T}}$. We denote the set of these blocks by ${\cal G} \subseteq {\cal C}$. The set ${\cal G}$ satisfies all claimed bounds. This finishes the proof.
\end{proof}

\subparagraph*{Parameter Guessing:}
We analyze our algorithms in terms of $n$ (the length of the strings), $T$ (the running time budget), $L$ (the length of the LCS), as well as $\lambda$ and $\mu$, defined as
\[ \lambda := \sum_{i,j} L_{ij} \qquad\text{and}\qquad \mu := \max_{\text{block seq.\ ${\cal S}$}} \sum_{(i,j) \in {\cal S}} M_{ij}, \]
where the maximum goes over all block sequences ${\cal S}$. Note that $\lambda$ is the total LCS length over all blocks and $\mu$ is the maximum total number of matching pairs along any block sequence.

The numbers $n$ and $T$ are part of the input, and we can assume to know $M$, since it can be computed in time $O(n)$. However, in order to set some parameters in our algorithms, it would be convenient to also know $L,\lambda,\mu$ up to constant factors (which seemingly is a contradiction, as our goal is to compute a polynomial-factor approximation of $L$). 

We therefore run our algorithms $\Oh(\log^3 n)$ times, once for each guess $\hat L = 2^i, \, \hat \lambda = 2^j$, and $\hat \mu = 2^k$. 
Then for at least one call we have $L/2 \le \hat L \le L$, $\lambda/2 \le \hat \lambda \le \lambda$, and $\mu/2 \le \hat \mu \le \mu$, that is, we know $L,\lambda,\mu$ up to constant factors. 
For this correct guess, we prove that our algorithms have the promised approximation guarantee and running time bound. 
For the wrong guesses, the approximation guarantee can fail, but we always ensure the upper bound $\widetilde L \le L$, by ensuring that the estimate corresponds to some common subsequence of $x$ and $y$. Hence, returning the maximum computed value~$\widetilde L$ over all guesses $\hat L, \hat \lambda, \hat \mu$ yields the promised approximation guarantee. For this reason, in the following we assume to know estimates $\hat L \approx L, \, \hat \lambda \approx \lambda, \, \hat \mu \approx \mu$ up to constant factors; we will only use them to set certain parameters. 

We remark that for the wrong guesses, not only the approximation guarantee but also the running time bound can fail, so we need to abort each of the $\Oh(\log^3 n)$ calls after time~$\tOh(T)$.


\subparagraph*{Diagonals:}
A \emph{diagonal} is a set of the form ${\cal D}_d = \{(i,j) \in [n/m]^2 \mid i-j=d\}$. Each diagonal is a block sequence, so we have $\sum_{(i,j) \in {\cal D}_d} M_{ij} \le \mu$. Note that there are $2n/m-1 < 2\sqrt{T}$ (non-empty) diagonals. Moreover, we have $\sum_d \sum_{(i,j) \in {\cal D}_d} M_{ij} = M$. This yields the inequality

\begin{align}  \label{eq:MmuT}
M < 2 \mu \sqrt{T}.
\end{align}

\subsection{Algorithm 2: Large \boldmath$L$, Large \boldmath$\mu$}
\label{sec:algotwo}

In this section we present Algorithm 2, which works well if $\mu$ is large, i.e., if some block sequence has a large total number of matching pairs. 
The algorithm makes use of the single symbol approximation that we designed in Lemma~\ref{lem:groupapxsinglesymbol}. This yields estimates $0 \le \widetilde L_{ij} \le L_{ij}$, over which we then perform dynamic programming to determine the maximum of $\sum_{(i,j) \in {\cal S}} \widetilde L_{ij}$ over all block sequences ${\cal S}$. 
(This is similar to \cite[Algorithm 3]{HajiaghayiSSS19}, but we obtain concentration in a wider regime, see Remark~\ref{rem:alg2} for a comparison.)

%

\begin{restatable}[Algorithm 2]{theorem}{thmb}
	\label{thm:alg2}
	We can compute in expected time $\tOh(T)$ an estimate $\widetilde L \le L$ that w.h.p.\ satisfies
	\[ \widetilde L = \Omega\Big(\frac{\mu \sqrt{T}}n \Big). \]
\end{restatable}


\begin{proof}
	Algorithm 2 works as follows.
	\begin{enumerate}
		\item[1.] Run Lemma~\ref{lem:groupapxsinglesymbol} with $q := \frac{M}{4T}$ to compute values $\widetilde L_{ij}$.
		\item[2.] Perform dynamic programming over $[n/m]^2$ to determine the maximum $\sum_{(i,j) \in {\cal S}} \widetilde L_{ij}$ over all block sequences ${\cal S}$. Output this maximum value $\widetilde L$. More precisely:
		\begin{itemize}
			\item Initialize $D[i,0] = D[0,i] = 0$ for any $0 \le i \le n/m$.
			\item For $i=1,\ldots,n/m$ and $j=1,\ldots,n/m$:  $D[i,j] = \max\big\{ \widetilde L_{ij} + D[i-1,j-1],\; D[i-1,j],\; D[i,j-1] \big\}$.
			\item Output $D[n/m,n/m]$.
		\end{itemize}
	\end{enumerate}
	
	We analyze this algorithm in the following.
	
	\subparagraph*{Upper Bound:}
	Since Lemma~\ref{lem:groupapxsinglesymbol} ensures $\widetilde L_{ij} \le L_{ij}$, the dynamic programming step ensures $\widetilde L \le L$. 
	
	\subparagraph*{Approximation Guarantee:}
	Let ${\cal S}$ be a block sequence achieving $\sum_{(i,j) \in {\cal S}} M_{ij} = \mu$.
	Step 2 computes an estimate $\widetilde L \ge \sum_{(i,j) \in {\cal S}} \widetilde L_{ij}$, and Lemma~\ref{lem:groupapxsinglesymbol} yields w.h.p.\ 
	\[ \widetilde L \ge \sum_{(i,j) \in {\cal S}} \widetilde L_{ij} \ge \sum_{(i,j) \in {\cal S}} \frac{M_{ij} - q}{16 m} = \frac{\mu}{16 m} - |{\cal S}|\frac{q}{16 m}. \]
	By the monotonicity property of block sequences, we have $|{\cal S}| \le n/m$. Using our definitions of $q = \frac M{4T}$ and $m=n/\sqrt{T}$ as well as inequality (\ref{eq:MmuT}), we obtain
	\[ |{\cal S}|\frac{q}{16 m} \le \frac{q n}{16 m^2} = \frac{M}{64n} \le \frac{\mu \sqrt{T}}{32n}. \]
	Plugging this into our bound for $\widetilde L$ yields 
	\[ \widetilde L \ge \frac {\mu \sqrt{T}}{16 n} - \frac {\mu \sqrt{T}}{32n} = \frac {\mu \sqrt{T}}{32n}. \]
	%
	%
	%
	%
	%
	
	\subparagraph*{Running Time:}
	For Step 1 note that Lemma~\ref{lem:groupapxsinglesymbol} runs in expected time $\tOh(n + M/q) = \tOh(T)$.
	Step 2 can be easily seen to run in time $\Oh((n/m)^2) = \Oh(T)$ by our choice of $m = n/\sqrt{T}$. This finishes the proof.
\end{proof}

\begin{remark} \label{rem:alg2}
	Our Algorithm 2 is similar to \cite[Algorithm 3]{HajiaghayiSSS19}, which works as follows. For each block $(i,j)$, their algorithm selects a random symbol $\sigma$ and uses the minimum of the frequencies $\fr_\sigma(x_i), \fr_\sigma(y_j)$ as the estimate $\widetilde L_{ij}$. It can be shown that this yields $\Ex[\widetilde L_{ij}] = M_{ij}/(2m)$, which is a similar lower bound as provided by Lemma~\ref{lem:groupapxsinglesymbol}, but only in expectation. The summation $\sum_{(i,j) \in {\cal S}} \widetilde L_{ij}$ over a block sequence ${\cal S}$ then allows to apply concentration inequalities to obtain a w.h.p.\ error guarantee, assuming $\mu \gg m^2$.
	
	However, in the regime $\mu \le m^2$ the value $\mu$ could be dominated by a single block with $M_{ij} \approx \mu$. In this case, we cannot hope to get concentration by summing over many blocks. Thus, picking a random symbol per block does not suffice to obtain a w.h.p.\ error guarantee. 
	
	Since our improved approximation ratio makes it necessary to use Algorithm 2 in the regime $\mu \ll m^2$, their algorithm is not sufficient in our context. Thus, we replace sampling a single symbol by our new Lemma~\ref{lem:groupapxsinglesymbol}.
\end{remark}

\subsection{Algorithm 3: Large \boldmath$L$, Small \boldmath$\mu$ and Large \boldmath$\lambda$}
\label{sec:algothree}

Our next algorithm works well if $\mu$ is small (i.e., every block sequence has a small total number of matching pairs) and $\lambda$ is large (i.e., on average every block has a large LCS). 

Let us start with the intuition.
The idea is to \emph{pick some diagonal ${\cal D}_d$ and run the basic approximation algorithm (Lemma~\ref{lem:extendedbasicapx}) with approximation ratio $\beta = \max\{1, \mu / T\}$ on each block along the diagonal}.
Since every diagonal is a block sequence, we have $\sum_{(i,j) \in {\cal D}_d} M_{ij} \le \mu$, which bounds the running time of this algorithm by $\tOh(n + \sum_{(i,j) \in {\cal D}_d} M_{ij}/\beta) = \tOh(T)$. 
Moreover, this algorithm produces an estimate $\widetilde L \le L$ that w.h.p.\ satisfies
\[\widetilde L \ge \sum_{(i,j) \in {\cal D}_d}  L_{ij}/\beta. \]
Since $\sum_d \sum_{(i,j) \in {\cal D}_d}  L_{ij} = \sum_{i,j} L_{ij} = \lambda$ and there are $\Oh(n/m)$ diagonals, on average a diagonal ${\cal D}_d$ satisfies $\sum_{(i,j) \in {\cal D}_d}  L_{ij} = \Omega(\lambda m/n) = \Omega(\lambda/ \sqrt{T})$. 
If we pick an above-average diagonal, then we obtain an estimate
\[\widetilde L \ge \sum_{(i,j) \in {\cal D}_d}  L_{ij}/\beta = \Omega\Big( \frac \lambda{\sqrt{T} \beta} \Big) = \Omega\Big(\min\Big\{ \frac \lambda{\sqrt{T}}, \frac{\lambda \sqrt{T}}\mu \Big\}\Big). \]
If $\lambda$ is large and $\mu$ is small, then this is a good estimate.

\medskip
The main difficulty in translating this idea to an actual algorithm is how to pick the diagonal. A natural approach is to pick a random diagonal, as then the \emph{expected} LCS sum of the diagonal is sufficiently large. However, in situations where the diagonal sums are highly unbalanced, so that $\lambda$ is dominated by very few diagonals that have a very large LCS sum, a random diagonal is unlikely to have an above-average LCS sum. In this situation, a random diagonal works only with negligible probability. 

Therefore, we need a sampling process that favors diagonals with large LCS sum. To this end, we first ``guess'' a value $g$ such that the sum $\lambda$ is dominated by summands $L_{ij} = \Theta(g)$. We call blocks $(i,j)$ with $L_{ij} = \Omega(g)$ \emph{good}. Next we sample a random good block $(i_0,j_0)$; for this we simply keep sampling random $i,j$ until we find a good block. Finally, we pick the diagonal ${\cal D}_d$ containing the ``seed'' block $(i_0,j_0)$ and run the above algorithm on this diagonal. This sampling procedure favors diagonals with large LCS sum, because such diagonals contain more good blocks $(i,j)$ to start from, and thus we are more likely to pick the ``seed'' $(i_0,j_0)$ in a diagonal with large LCS sum. This yields the following result.

\begin{restatable}[Algorithm 3]{theorem}{thmc}
	\label{thm:alg3}
	We can compute in expected time $\tOh(T)$ an estimate $\widetilde L \le L$ that w.h.p.\ satisfies
	\[\widetilde L = \widetilde \Omega\Big(\min\Big\{ \frac \lambda{\sqrt{T}}, \frac{\lambda \sqrt{T}}\mu \Big\}\Big). \]
\end{restatable}

\begin{proof}
	Note that the theorem statement is trivial if $\lambda \le \sqrt{T}$. Indeed, in time $O(n)$ we can compute $M = M(x,y)$. If $M = 0$ then $L = \lambda = 0$ and we return $\widetilde L = 0$. If $M \ge 1$, then we return $\widetilde L = 1$. This ensures $\widetilde L \le L$, since any matching pair gives a common subsequence of length 1. Moreover, in case $\lambda \le \sqrt{T}$ the returned value $\widetilde L = 1$ satisfies the approximation guarantee $\widetilde L = \Omega(\lambda/\sqrt{T})$. Therefore, we can assume
	
	\begin{align} \label{eq:assumelambdaT}
	\lambda > \sqrt{T}.
	\end{align}
	
	Algorithm 3 repeats the following procedure $\Oh(\log n)$ times to boost its success probability.
	\begin{enumerate}
		\item[1.] Repeat the following for $g$ being any power of two with $\max\{1,\hat \lambda/(4T)\} \le g \le m$:
		\begin{enumerate}
			\item[2.] \emph{Sampling a good block:}
			Pick a random set of blocks ${\cal R} \subseteq [n/m]^2$ of size $\Oh\big( (g T/\hat \lambda) \log^2 n\big)$. 
			For each block $(i,j) \in {\cal R}$, test whether $L_{ij} \ge g$ using our basic decision algorithm (Lemma~\ref{lem:basicdec}). If no test was successful, then set $\widetilde L(g) = 0$ and continue with the next value of $g$. Otherwise, pick a random successfully tested block $(i_0,j_0)$ and proceed to Step 3.
			%
			\item[3.] \emph{Approximating along a diagonal:}
			Let ${\cal D}$ be the diagonal containing the block $(i_0,j_0)$. For each $(i,j) \in {\cal D}$: Run our basic approximation algorithm (Lemma~\ref{lem:extendedbasicapx}) with approximation ratio $\beta = \max\{1, \hat \mu / T\}$ on $x_i,y_j$ to obtain an estimate $\widetilde L_{ij}$. Finally, $\widetilde L(g) = \sum_{(i,j) \in {\cal D}} \widetilde L_{ij}$ is the result of iteration $g$. 
		\end{enumerate}
		\item[4.] Return $\widetilde L = \max_g \widetilde L(g)$.
	\end{enumerate}
	
	\subparagraph*{Upper Bound:}
	Again it is easy to see that $\widetilde L \le L$, since Lemma~\ref{lem:extendedbasicapx} yields $\widetilde L_{ij} \le L_{ij}$.
	
	\subparagraph*{Approximation Guarantee:}
	Let ${\cal B}_g$ be the set of all blocks $(i,j)$ with $g \le L_{ij} \le 2g$. 
	
	\begin{claim} \label{cla:gBg}
		If $\lambda/2 \le \hat \lambda \le \lambda$ then for some power of two $g$ with $\max\{1,\hat \lambda/(4T)\} \le g \le m$ we have 
		
		\begin{align} \label{eq:gBg}
		g \cdot |{\cal B}_g| = \Omega( \lambda / \log m).
		\end{align}
		
	\end{claim}
	\begin{proof}
		Write $G$ for the set of all powers of two $g$ with $\max\{1,\hat \lambda/(4T)\} \le g \le m$.
		Note that blocks $(i,j)$ with $L_{ij} \le \frac \lambda{2T}$ in total contribute at most $\lambda/2$ to $\lambda = \sum_{i,i} L_{ij}$, since the total number of blocks is $(n/m)^2 = T$. Hence, the blocks with $L_{ij} > \frac \lambda{2T}$ contribute at least $\lambda/2$, that is,
		\[ \frac \lambda 2 \le \sum_{\substack{i,j \\ L_{ij} > \lambda/(2T)}} L_{ij}. \]
		Note that the sets ${\cal B}_g$ for powers of two $g \ge \max\{1, \lambda/(4T)\} \ge \max\{1, \hat \lambda/(4T)\}$ cover all blocks with $L_{ij} > \frac \lambda{2T}$. Moreover, the sets ${\cal B}_g$ are empty for $g > m$. Therefore, the blocks with $L_{ij} > \frac \lambda{2T}$ are covered by the sets ${\cal B}_g$ with $g \in G$, that is,
		\[ \frac \lambda 2 \le \sum_{\substack{i,j \\ L_{ij} > \lambda/(2T)}} L_{ij} \le \sum_{g \in G} \sum_{(i,j) \in {\cal B}_g} L_{ij} \le \sum_{g \in G} 2 g |{\cal B}_g|. \]
		If for all $g$ appearing in the sum on the right hand side we would have $g \cdot |{\cal B}_g| < \lambda / (4 \log m + 4)$ then the right hand side would be less than $\lambda/2$, so we would obtain a contradiction. This proves the claim.
	\end{proof}
	
	
	In the following we focus on an iteration of Step 1 in which we pick a value of $g$ as promised by Claim~\ref{cla:gBg}.
	
	
	We call a block $(i,j)$ \emph{good} if $L_{ij} \ge g$, and \emph{bad} otherwise. Note that any $(i,j) \in {\cal B}_g$ is good, but not every good block is in ${\cal B}_g$.
	In Step 2, we claim that the set ${\cal R}$ w.h.p.\ contains at least one good block, assuming that our guess $\hat \lambda$ is correct up to constant factors. Indeed, since the set ${\cal B}_g$ is a subset of the good blocks, the probability that $\Theta((gT/\lambda) \log^2 n)$ sampled blocks do not contain any good block is at most
	\[ \bigg(1 - \frac{|{\cal B}_g|}{(n/m)^2} \bigg)^{\Theta((gT/\lambda) \log^2 n)} \;\;\stackrel{(\ref{eq:gBg})}{\le}\;\; \bigg( 1 - \frac{\lambda}{g T \log m} \bigg)^{\Theta((gT/\lambda) \log^2 n)} \;\;\le\;\; \exp(-\Theta(\log n)), \]
	which is negligible.
	For any bad block $(i,j) \in {\cal R}$ the test $L_{ij} \ge g$ is unsuccessful, as Lemma~\ref{lem:basicdec} has no false positives. For any good block $(i,j) \in {\cal R}$ w.h.p.\ the test is successful, and w.h.p.\ there is at least one good block in ${\cal R}$. It follows that w.h.p.\ Step 2 finds a good block $(i_0,j_0)$ and proceeds to Step 3. Observe that $(i_0,j_0)$ is chosen uniformly at random from all good blocks.
	
	
	We call a diagonal \emph{good} if it contains at least $\frac {|{\cal B}_g| m}{4n}$ good blocks, and \emph{bad} otherwise. Since there are $< 2n/m$ non-empty diagonals, the number of good blocks contained in bad diagonals is at most $|{\cal B}_g|/2$, which is at most half of all good blocks. Therefore, at least half of all good blocks are contained in good diagonals. It follows that the uniformly random good block $(i_0,j_0)$ lies in a good diagonal with probability at least $1/2$. 
	
	Hence, with probability at least $1/2 - o(1)$ the diagonal ${\cal D}$ considered in Step 3 is good, that is, it contains at least $\frac {|{\cal B}_g| m}{4n}$ blocks $(i,j)$ with $L_{ij} \ge g$. 
	Since the approximations $\widetilde L_{ij}$ computed in Step 3 w.h.p.\ satisfy $\widetilde L_{ij} \ge L_{ij}/\beta$, we obtain
	\[ \widetilde L \ge \widetilde L(g) \ge \frac {|{\cal B}_g| m}{4n} \cdot \frac g \beta. \]
	Inequality (\ref{eq:gBg}) and the definitions $m = n/\sqrt{T}$ and $\beta = \max\{1, \hat \mu/T\}$ now yield
	\[ \widetilde L = \tOmega\Big(\min\Big\{ \frac{\lambda}{\sqrt{T}}, \frac{\lambda \sqrt{T}}{\hat \mu} \Big\}\Big). \]
	If our guess $\hat \mu \approx \mu$ is correct up to a constant factor, then this yields the claimed approximation guarantee. Returning the maximum over $\Oh(\log n)$ independent repetitions of this algorithm improves the success probability from $1/2-o(1)$ to w.h.p.
	
	\subparagraph*{Running Time:}
	By Lemma~\ref{lem:basicdec}, the test $L_{ij} \ge g$ runs in expected time $\tOh((m+M_{ij})L_{ij}/g + m/g) = \tOh(m^2 L_{ij}/g + m/g)$. Note that in expectation for random $i,j$ we have $\Ex[L_{ij}] = \lambda/(n/m)^2 = \lambda/T$. Therefore, the expected running time of one test is bounded by $\tOh\big(\frac{m^2 \lambda}{gT} + m/g \big)$. As Step 2 performs $\Oh\big( (g T/\hat \lambda) \log^2 n\big)$ such tests, its expected running time is $\tOh(m^2 + mT/\lambda)$, assuming that our guess $\hat \lambda \approx \lambda$ is correct up to a constant factor. We now use $m^2 = n^2/T \le T$ from $n \le T$ and $\lambda \ge \sqrt{T} \ge n/\sqrt{T} = m$ from (\ref{eq:assumelambdaT}) and $n \le T$, to bound the expected running time of Step 2 by $\tOh(T)$. 
	
	For Step 3, the expected running time is $\tOh(n + \sum_{(i,j) \in {\cal D}} M_{ij}/\beta)$. Since ${\cal D}$ is a block sequence, we have $\sum_{(i,j) \in {\cal D}} M_{ij} \le \mu$. Using $\beta \ge \hat \mu/T = \Omega(\mu/T)$ (if our guess $\hat \mu \approx \mu$ is correct up to a constant factor) we can bound the expected time by $\tOh(n + T) = \tOh(T)$.
	
	Over the $\Oh(\log n)$ iterations of Step 1 and the $\Oh(\log n)$ repetitions for boosting the success probability, the expected running time is still $\tOh(T)$.
\end{proof}

\subsection{Algorithm 4: Large \boldmath$L$, Small \boldmath$\mu$, and Small \boldmath$\lambda$}
\label{sec:algofour}

Our next algorithm works well if $\mu$ is small (i.e., every block sequence has a small total number of matching pairs), $\lambda$ is small (i.e., on average every block has a small LCS), and $L$ is large (i.e., there is a long LCS).
The goal of this algorithm is to detect a sufficiently large random subset of the block sequence ${\cal G}$ from Lemma~\ref{lem:blockaligned}. 
To this end, we first sample a random set of blocks ${\cal R}$ containing each block $(i,j) \in [n/m]^2$ with probability $p$. Then we use our basic decision algorithm to detect the blocks $(i,j) \in {\cal R}$ with $L_{ij}  \ge \frac {\hat L}{4\sqrt{T}}$, and for these blocks we set $\widetilde L_{ij} = \frac {\hat L}{4\sqrt{T}}$, while for the remaining blocks we set $\widetilde L_{ij} = 0$. Finally, we perform dynamic programming to determine the maximum $\sum_{(i,j) \in {\cal S}} \widetilde L_{ij}$ over all block sequences ${\cal S}$.

Observe that for each block in ${\cal G} \cap {\cal R}$ this algorithm sets $\widetilde L_{ij} = \frac {\hat L}{4\sqrt{T}}$, so it detects a random subset of~${\cal G}$.
We thus obtain a $p$-fraction of the LCS guaranteed by the block sequence ${\cal G}$. 

Note that in this algorithm we may focus on blocks with $M_{ij} = \Oh(\frac{\mu n}{L \sqrt{T}})$, since this holds for all blocks in~${\cal G}$. Moreover, since $\lambda$ is small, most blocks outside of~${\cal G}$ have small LCS $L_{ij}$. These bounds on $L_{ij}$ and $M_{ij}$ for the considered blocks allow us to bound the running time of the basic decision algorithm.
We elaborate this algorithm in the following theorem.

\begin{restatable}[Algorithm 4]{theorem}{thmd}
	\label{thm:alg4}
	We can compute in expected time $\tOh(T)$ an estimate $\widetilde L \le L$ that w.h.p.\ satisfies
	\[ \widetilde L = \Omega\Big(\min\Big\{ \frac{L^3}{n^2}, \frac{L^3 T}{\lambda n^2}, \frac{L^4 T}{\lambda \mu n^2} \Big\} \Big), \;\; \text{assuming that }\;\;
	\frac{L^2 T^{0.5}}{n^2}, \frac{L^2 T^{1.5}}{\lambda n^2}, \frac{L^3 T^{1.5}}{\lambda \mu n^2} = n^{\Omega(1)}.\]
\end{restatable}

\begin{proof}
	Algorithm 4 works as follows.
	\begin{enumerate}
		\item[1.] Run Lemma~\ref{lem:countmatchingpairs} with $q := \frac{M}{T}$ to compute values $\widetilde M_{ij}$. Initialize $\widetilde L_{ij} = 0$ for all $i,j$.
		\item[2.] Run the preprocessing of the basic decision algorithm (Lemma~\ref{lem:basicdec}) on each string $y_j$.
		\item[3.] Sample a set ${\cal R} \subseteq [n/m]^2$ by including each block $(i,j)$ independently with probability 
		\[ p := \min\Big\{ \frac{\hat L}n, \frac{\hat L T}{\hat \lambda n}, \frac{\hat L^2 T}{\hat \lambda \hat \mu n} \Big\}. \]
		\item[4.] For each $(i,j) \in {\cal R}$ with $\widetilde M_{ij} \le 64 \hat \mu n/(\hat L \sqrt{T})$: Run the query of the basic decision algorithm (Lemma~\ref{lem:basicdec}) to test whether $L_{ij} \ge \frac{\hat L}{4\sqrt{T}}$. If this test is successful then set $\widetilde L_{ij} := \frac{\hat L}{4\sqrt{T}}$.
		\item[5.] Perform dynamic programming over $[n/m]^2$ to determine the maximum $\sum_{(i,j) \in {\cal S}} \widetilde L_{ij}$ over all block sequences ${\cal S}$. Output this maximum value $\widetilde L$.
	\end{enumerate}
	
	\subparagraph*{Upper Bound:}  
	Since Lemma~\ref{lem:basicdec} has no false positives, we ensure $\widetilde L_{ij} \le L_{ij}$ and thus $\widetilde L \le L$. 
	
	
	\subparagraph*{Approximation Guarantee:}
	The values $\widetilde M_{ij}$ computed in Step 1 w.h.p.\ satisfy $M_{ij}/8 - q \le \widetilde M_{ij} \le 4 M_{ij}$.
	For all blocks $(i,j) \in {\cal G}$ we have $M_{ij} \le \frac{8 \mu n}{L \sqrt{T}}$ (by Lemma~\ref{lem:blockaligned}) and thus w.h.p.\ $\widetilde M_{ij} \le \frac{32 \mu n}{L \sqrt{T}}$. 
	We may assume that our guesses $\hat L, \hat \mu$ satisfy $L/2 \le \hat L \le L$ and $\mu/2 \le \hat \mu \le \mu$; then we obtain $\widetilde M_{ij} \le \frac{64 \hat \mu n}{\hat L \sqrt{T}}$. 
	
	Therefore, each block in ${\cal G}\cap {\cal R}$ satisfies the property checked in Step 4, that is, for each such block we run the basic decision algorithm. Since for each $(i,j) \in {\cal G}$ we have $L_{ij} \ge \frac L{4\sqrt{T}} \ge \frac{\hat L}{4 \sqrt{T}}$, in Step 4 for each block in ${\cal G}\cap {\cal R}$ w.h.p.\ we obtain an estimate $\widetilde L_{ij} = \frac{\hat L}{4 \sqrt{T}}$. Since ${\cal G}$ is a block sequence, also ${\cal G}\cap {\cal R}$ is a block sequence, and thus the dynamic programming in Step 5 returns an estimate of 
	\[ \widetilde L \ge \sum_{(i,j) \in {\cal G}\cap {\cal R}} \widetilde L_{ij} = |{\cal G}\cap {\cal R}| \cdot \frac{\hat L}{4 \sqrt{T}}. \]
	Note that the size $|{\cal G} \cap {\cal R}|$ is distributed as a binomial random variable $\textup{Bin}(|{\cal G}|, p)$, with expectation $p |{\cal G}|$. Assuming that our guesses $\hat L, \hat \lambda, \hat \mu$ are correct up to constant factors, we have
	\[ p |{\cal G}| = \Omega\Big( \min\Big\{ \frac Ln, \frac{L T}{\lambda n}, \frac{L^2 T}{\lambda \mu n} \Big\} \cdot \frac{L \sqrt{T}}{n} \Big) = \Omega\Big( \min\Big( \frac{L^2 T^{0.5}}{n^2}, \frac{L^2 T^{1.5}}{\lambda n^2}, \frac{L^3 T^{1.5}}{\lambda \mu n^2} \Big) \Big) = n^{\Omega(1)}, \]
	by the assumption in the theorem statement.
	By Chernoff bound, we have
	\[ \Pr[ |{\cal G} \cap {\cal R}| < p |{\cal G}| / 2 ] \le \exp(-p |{\cal G}| / 8) = \exp(-n^{\Omega(1)}), \]
	and thus w.h.p.\ we have $|{\cal G} \cap {\cal R}| \ge p |{\cal G}| / 2$. Plugging this into our lower bound for $\widetilde L$ yields w.h.p.\
	\[ \widetilde L \ge \frac{p |{\cal G}| \hat L}{8 \sqrt{T}} = \Omega\Big(\min\Big\{ \frac{L^3}{n^2}, \frac{L^3 T}{\lambda n^2}, \frac{L^4 T}{\lambda \mu n^2} \Big\} \Big), \]
	assuming that our guesses $\hat L, \hat \lambda, \hat \mu$ are correct up to constant factors. This shows the claimed lower bound.
	
	%
	
	\subparagraph*{Running Time:}
	The expected running time of Step 1 is $\tOh(n + M/q) = \tOh(T)$ since $q = \frac M{T}$. Step 2 runs in time $\tOh(\sum_j |y_j|) = \tOh(n)$. Steps 3 and 5 take time $\Oh((n/m)^2) = \Oh(T)$. 
	In the remainder we show that Step 4 also runs in expected time $\tOh(T)$, assuming that our guesses $\hat L, \hat \lambda, \hat \mu$ are correct up to constant factors. Recall that w.h.p.\ $M_{ij}/8 - q \le \widetilde M_{ij} \le 4 M_{ij}$; we condition on this event in the following.\footnote{In the error event we bound the running time of Step 4 by $\Oh(n^2)$. This has a negligible contribution to the expected running time of Step 4.}
	Then $\widetilde M_{ij} = \Oh(\frac{\mu n}{L \sqrt{T}})$ implies $M_{ij} = \Oh(q + \frac{\mu n}{L \sqrt{T}})$. 
	Using inequality (\ref{eq:MmuT}) we have $q = \frac M{T} \le \frac {2\mu}{\sqrt{T}} \le \frac{2\mu n}{L \sqrt{T}}$, so $M_{ij} = \Oh(\frac{\mu n}{L \sqrt{T}})$. 
	Since only blocks $(i,j)$ with $\widetilde M_{ij} = \Oh(\frac{\mu n}{L \sqrt{T}})$ are tested, each invocation of the basic approximation algorithm in Step 4 runs in expected time $\tOh((|x_i| + M_{ij})L_{ij} \sqrt{T}/L + |x_i| \sqrt{T}/L) = \tOh((m + \frac{\mu n}{L \sqrt{T}}) L_{ij} \sqrt{T}/L + m \sqrt{T}/L)$. 
	Since each block is tested with probability at most $p$, Step 4 has an expected running time of 
	\begin{align*} \Oh\Big( \sum_{i,j} p \cdot \Big(\Big(m + \frac{\mu n}{L \sqrt{T}}\Big) L_{ij} \frac{\sqrt{T}}{L} + \frac{m \sqrt{T}}L\Big) \Big)
	&= \Oh\Big( \Big(\frac n{\sqrt{T}} + \frac{\mu n}{L \sqrt{T}}\Big) \frac{p \lambda \sqrt{T}}{L} + \frac {pnT}L \Big) \\
	&= \tOh\Big(\frac{p \lambda n}L + \frac{p \lambda \mu n}{L^2} + \frac {pnT}L \Big). 
	\end{align*}
	Note that our choice of $p = \Theta(\min\{\frac Ln, \frac{LT}{\lambda n}, \frac{L^2T}{\lambda \mu n} \})$ ensures that this running time is $\tOh(T)$.
\end{proof}

\subsection{Combining the Algorithms}
\label{sec:combine}

Now we combine Algorithms 1-4. We show that w.h.p.\ at least one of these algorithms computes an estimate $\widetilde L = \tOmega( L T^{0.4} / n^{0.8} )$. In other words, the combined algorithm has approximation ratio $\tOh(n^{0.8} / T^{0.4})$ with a running time budget of $\tOh(T)$ and thus prove Theorem \ref{thm:maintechnical}.

\thmg*

\subparagraph*{Algorithm 1:} Recall from Theorem~\ref{thm:alg1} that Algorithm 1 w.h.p.\ returns $\widetilde L \ge \min\{L, \sqrt{LT/n}\}$. If $\widetilde L \ge L$ then we solved the problem exactly, so we only need to consider the case $\widetilde L \ge \sqrt{LT/n} = L \sqrt{T/(Ln)}$. Assuming that $L \le n^{0.6} T^{0.2}$, we obtain the claimed approximation guarantee $\widetilde L \ge L T^{0.4} / n^{0.8}$. Hence, from now on we can assume

\begin{align} \label{eq:assumpL}
L > n^{0.6} T^{0.2}. 
\end{align}

\subparagraph*{Algorithm 2:} Recall from Theorem~\ref{thm:alg2} that Algorithm 2 w.h.p.\ returns $\widetilde L = \Omega(\mu \sqrt{T}/ n )$. Assuming that $\mu \ge L n^{0.2} / T^{0.1}$, we obtain the claimed approximation guarantee $\widetilde L = \Omega( L T^{0.4} / n^{0.8} )$. Hence, from now on we can assume

\begin{align}  \label{eq:assumpMu}
\mu < \frac{L n^{0.2} }{ T^{0.1}}.
\end{align}

\subparagraph*{Algorithm 3:} Recall from Theorem~\ref{thm:alg3} that Algorithm 3 w.h.p.\ returns $\widetilde L = \tOmega\big(\min\big\{ \frac \lambda{\sqrt{T}}, \frac {\lambda \sqrt{T}} \mu \big\} \big)$. Assuming that $\lambda \ge L^2 T^{0.7} / n^{1.4}$, we have 
\[ \widetilde L = \tOmega\Big(\min\Big\{ \frac {L^2 T^{0.2}}{ n^{1.4}}, \frac{L^2 T^{1.2} }{\mu n^{1.4}} \Big\}\Big). \] Bounding one factor $L$ by (\ref{eq:assumpL}) and $\mu$ by (\ref{eq:assumpMu}) yields 
\[ \widetilde L = \tOmega\Big(\min\Big\{ \frac{L T^{0.4} }{ n^{0.8}}, \frac{L T^{1.3}}{ n^{1.6}} \Big\}\Big). \]
It remains to see that $T^{1.3} / n^{1.6} \ge T^{0.4} / n^{0.8}$, which is equivalent to $T^{0.9} \ge n^{0.8}$ and thus follows from $T \ge n$. Hence, Algorithm 3 satisfies the claimed approximation guarantee $\widetilde L = \tOmega( L T^{0.4} / n^{0.8} )$, under our assumption on $\lambda$.
We can thus from now on assume 

\begin{align} \label{eq:assumpLa}
\lambda < \frac{L^2 T^{0.7}}{ n^{1.4}}.
\end{align}

\subparagraph*{Algorithm 4:} We first verify that (\ref{eq:assumpL}), (\ref{eq:assumpMu}), and (\ref{eq:assumpLa}) imply the assumptions of Theorem~\ref{thm:alg4}:
\[ \frac{L^2 T^{0.5}}{n^2}, \frac{L^2 T^{1.5}}{\lambda n^2}, \frac{L^3 T^{1.5}}{\lambda \mu n^2} = n^{\Omega(1)}. \]
Indeed, we have
\[ \frac{L^2 T^{0.5}}{n^2} \;\stackrel{(\ref{eq:assumpL})}{>}\; \frac{T^{0.9} }{ n^{0.8}} \;\ge\; n^{0.1} \;=\; n^{\Omega(1)}, \]
where we used $T \ge n$.
Similarly, we have
\[ \frac{L^2T^{1.5}}{\lambda n^2} \;\stackrel{(\ref{eq:assumpLa})}{>}\; \frac{T^{0.8} }{ n^{0.6}} \;\ge\; n^{0.2} \;=\; n^{\Omega(1)}, \]
and
\[ \frac{L^3T^{1.5}}{\lambda \mu n^2} \;\stackrel{(\ref{eq:assumpLa})}{>}\; \frac{L T^{0.8}}{\mu n^{0.6}} \;\stackrel{(\ref{eq:assumpMu})}{>}\; \frac{T^{0.9}}{n^{0.8}} \;\ge\; n^{0.1} \;=\; n^{\Omega(1)}. \]

As these assumptions hold, Theorem~\ref{thm:alg4} shows that Algorithm 4 w.h.p.\ returns $\widetilde L = \Omega(\min\{ \frac{L^3}{n^2}, \frac{L^3 T}{\lambda n^2}, \frac{L^4 T}{\lambda \mu n^2} \})$. We verify that assuming (\ref{eq:assumpL}), (\ref{eq:assumpMu}), and (\ref{eq:assumpLa}) this yields the claimed approximation guarantee. 
Indeed, we have
\[ \frac{L^3}{n^2} \;\stackrel{(\ref{eq:assumpL})}{>}\; \frac{L T^{0.4}}{n^{0.8}}. \]
Similarly, we have
\[ \frac{L^3 T}{\lambda n^2} \;\stackrel{(\ref{eq:assumpLa})}{>}\; \frac{L T^{0.3}}{n^{0.6}} \ge \frac{L T^{0.4}}{n^{0.8}}, \]
since $T \le n^2$.
Finally, we have
\[ \frac{L^4 T}{\lambda \mu n^2} \;\stackrel{(\ref{eq:assumpMu})}{>}\; \frac{L^3 T^{1.1}}{\lambda n^{2.2}} \;\stackrel{(\ref{eq:assumpLa})}{>}\; \frac{L T^{0.4}}{n^{0.8}}. \]
In all cases we obtain a lower bound of $\widetilde L = \tOmega(L T^{0.4} / n^{0.8})$. 

\medskip
In summary, we proved that w.h.p.\ at least one of the Algorithms 1-4 satisfies the desired approximation guarantee of $\widetilde L = \tOmega( L T^{0.4} / n^{0.8} )$. This concludes the proof of Theorem~\ref{thm:maintechnical}.

\paragraph{Acknowledgements.}
We thank an anonymous reviewer for
	suggesting how to turn the near-linear-time algorithm that we obtained
	in a previous version of this paper into a linear-time algorithm.

\bibliography{main} 

\appendix

\section{Hunt and Szymanski's LCS Algorithm}
\label{app:prelim}

In this section we provide a proof sketch of Theorem~\ref{thm:hs77}.

\thmhs*

\begin{proof}
  Since this algorithm is typically stated as running in time $\tOh(n + M)$, here for convenience we sketch the algorithm and show how to split it into preprocessing and query phase.
  
  In the preprocessing phase, given string $y$ we compute the set $\Sigma(y)$ of symbols occuring in $y$, and for each $\sigma \in \Sigma(y)$ we compute a sorted array $A_\sigma$ containing the positions at which $\sigma$ appears in $y$.
  
  In the query phase, we are given a string $x$ and a preprocessed string $y$. The algorithm builds a dynamic programming table $T$ of length $|x|+1$, maintaining the following invariant: After the $i$-th round, $T[k]$ stores the minimum $j$ such that $L(x[1..i],y[1..j]) = k$ (or $\infty$ if not such $j$ exists).
  
  Initially, corresponding to round $i=0$, the table $T$ is computed by setting $T[0] = 0$ and $T[k] = \infty$ for any $k \in [|x|]$.
  Then in round $i$ the goal is to match $x[i]$. Therefore, we iterate over all $j \in A_{x[i]}$ in decreasing order; note that this enumerates all positions $j$ in $y$ that match $x[i]$. For each such $j$, we binary search for a value of $k$ with $T[k-1] < j \le T[k]$, and we set $T[k] = j$. This can be seen to maintain the invariant. In the end, the largest $k$ with $T[k] \ne \infty$ is equal to $L(x,y)$. In pseudocode, this algorithm does the following.

\begin{enumerate}
\item[1.] Preprocessing: Compute for each symbol $\sigma \in \Sigma(y)$ an array $A_\sigma$ listing the positions at which $\sigma$ appears in $y$, in sorted order.
\item[2.] Initialization of $T$: $T[0] \gets 0$, $T[k] \gets \infty$ for any $k \in [|x|]$.
\item[3.] For each $i$ in $[|x|]$: For each $j$ in $A_{x[i]}$ in decreasing order:
  \begin{enumerate}
  \item[4.] Find $k$ such that $T[k-1] < j \le T[k]$, and set $T[k] = j$.
  \end{enumerate}
\item[5.] Return the largest $k$ such that $T[k] \neq \infty$.
\end{enumerate}
It is easy to see that the preprocessing can be implemented in time $O(|y| \log n)$ and the rest of the algorithm runs in time $O((|x| + M) \log n)$.
\end{proof}

\input{pseudocodes.tex}

\end{document}

%% file: pseudocodes.tex
\newpage
\section{Pseudocodes}
\label{app:pseudocodes}

In this section we provide pseudocode for our algorithms.


\subsection{New Basic Tools}

We start by providing the pseudocodes for the basic tools that are required to design our main algorithm for approximating the LCS of the input strings.

\setcounter{algorithm}{4}

\subsubsection{Basic Approximation Algorithm}
\begin{algorithm} [H]
	\begin{algorithmic}[1]
		\REQUIRE String $x$ and preprocessed string $y$ of length $n$; parameter $\beta\ge 1$ 
		
		\ENSURE Computes subsequence $\widetilde{L}$ of $x,y$ satisfying $\widetilde{L}>\frac{L}{\beta}-1$ w.h.p.
		
		\STATE $x'\gets$ an empty string
		
		\IF{$\beta\ge \frac{1}{8c\log n}$}
		\STATE $\widetilde{L}\gets$ compute $LCS(x,y)$ using the query phase of Theorem \ref{thm:hs77}
		\ELSE
		
		\STATE $p\gets \frac{8c\log n}{\beta}$
		
		\STATE $q\gets Geo(p)$

		\WHILE{$q\le n$}
		
		\STATE append $x[q]$ to $x'$
		
		\STATE $q\gets q+Geo(p)$
		
		\ENDWHILE
		
		\STATE $\widetilde{L}\gets$ compute $LCS(x',y)$ using the query phase of Theorem \ref{thm:hs77}
		
		
		\ENDIF
		
		\RETURN $\widetilde{L}$
		
		\caption{{[Lemma \ref{lem:basicapx}] Basic Approximation Algorithm: Query} $(x,y,\beta)$}
		\label{alg:tool1}
	\end{algorithmic}
\end{algorithm}

The function $Geo(p)$ samples a geometric random variable from $\{1,2,3,...\}$ in constant time, see e.g. \cite{BringmannF13}.


\begin{algorithm}[H]
	\begin{algorithmic}[1]
		\REQUIRE String $x$ and preprocessed string $y$ of length $n$; parameter $\beta\ge 1$ 
		
		\ENSURE Computes subsequence $\widetilde{L}$ of $x,y$ satisfying $\widetilde{L}\ge\frac{L}{\beta}$ w.h.p.
		
		\STATE $\widetilde{L}\gets$ $Query(x,y,\beta)$ of Algorithm \ref{alg:tool1}
		
		\STATE $M\gets$ result of Algorithm \ref{alg:matching} on $(x,y)$
		
		\IF{$M>0$}
		\STATE $\widetilde{L}\gets \max\{\widetilde{L},1\}$
		
		\ENDIF
		
		\RETURN $\widetilde{L}$
		
		\caption{{[Lemma \ref{lem:extendedbasicapx}] Generalised Basic Approximation Algorithm: Query} $(x,y,\beta)$}
		\label{alg:tool2}
	\end{algorithmic}
\end{algorithm}


\begin{algorithm}[H]
	\begin{algorithmic}[1]
		\REQUIRE String $x$ and preprocessed string $y$ of length $n$; parameter $1\le \ell\le n$ 
		
		\ENSURE Decides whether $L\ge \ell$ w.h.p.
		
		\FOR{$\beta=n,n/2,\dots,1$}
		
		\STATE $\widetilde{L}\gets$ $Query(x,y,\beta)$ of Algorithm \ref{alg:tool1}
		
		\IF{$\widetilde{L}\le \ell/\beta-1$}
		
		\RETURN $0$
		
		\ENDIF
		
		\IF{$\widetilde{L}\ge \ell$}
		
		\RETURN $1$
		
		\ENDIF
		
		\ENDFOR

		\caption{{[Lemma \ref{lem:basicdec}] Basic Decision Algorithm: Query} $(x,y,\ell)$}
		\label{alg:tool3}
	\end{algorithmic}
\end{algorithm}

\newpage
\subsubsection{Approximating the Number of Matching Pairs}


\begin{algorithm} 
	\begin{algorithmic}[1]
		\REQUIRE Strings $x,y$ of length $n$ 
		
		\ENSURE Computes the number of matching pairs $M=M(x,y)$
		
		\STATE $M\gets 0$
		
		\STATE $\Sigma\gets$ Compute the set of symbols appearing in $x,y$

		


	
		
		\STATE Compute the frequencies $\fr_\sigma(x)$ for all $\sigma \in \Sigma$ by one linear scan over $x$
		
		\STATE Compute the frequencies $\fr_\sigma(y)$ for all $\sigma \in \Sigma$ by one linear scan over $y$
		
%
%
		
		\FOR{$\sigma\in \Sigma$}
		
		\STATE $M\gets M+\fr_\sigma(x)\cdot \fr_\sigma(y)$

		\ENDFOR

		\RETURN $M$
		\caption{{Counting Matching Pairs} $(x,y)$}
		\label{alg:matching}
	\end{algorithmic}
\end{algorithm}


\begin{algorithm}[H]
	\begin{algorithmic}[1]
		\REQUIRE Strings $x_1,\dots,x_{n/m},y_1,\dots,y_{n/m}$ of length $m$ 
		
		\ENSURE Constructs a three-layered graph $G$
		
			\STATE $\Sigma\gets$ Compute the set of symbols appearing in $x_1,\dots,x_{n/m},y_1,\dots,y_{n/m}$ 
			
		
		\FOR{$i\in[n/m]$}
		
		\STATE Compute the frequencies $\fr_\sigma(x_i)$ for all $\sigma \in \Sigma$ by one linear scan over $x_i$
		
		\STATE Compute the frequencies $\fr_\sigma(y_i)$ for all $\sigma \in \Sigma$ by one linear scan over $y_i$
		
%
%
%

		\ENDFOR

		\STATE $L\gets \{1,\dots,n/m\}$
		
		\STATE $R\gets \{1,\dots,n/m\}$
		
		\STATE $U\gets \{(\sigma,\ell,r) \mid \sigma\in \Sigma,0\le \ell,r\le \log m\}$
		
		\STATE $V\gets L\cup U\cup R$
		
		\STATE $E_1\gets \{(i,(\sigma,\ell,r)) \in L \times U \mid \fr_\sigma(x_i)\in [2^\ell,2^{\ell+1}) \}$
		
		\STATE $E_2\gets \{(j, (\sigma,\ell,r)) \in R \times U \mid \fr_\sigma(y_j)\in [2^r,2^{r+1}) \}$
		
		\STATE $E\gets E_1\cup E_2$
		
		\RETURN $G = (V,E)$
		\caption{{Graph Construction} $(x_1,\dots,x_{n/m},y_1,\dots,y_{n/m})$}
		\label{alg:graph}
	\end{algorithmic}
\end{algorithm}



\begin{algorithm}[H]
	\begin{algorithmic}[1]
		\REQUIRE Strings $x_1,\dots,x_{n/m},y_1,\dots,y_{n/m}$ of length $m$ over alphabet $\Sigma$; a parameter $q>0$
		
		\ENSURE For each $i,j\in[n/m]$ computes $\widetilde M_{ij}$ satisfying $M_{ij}/8 - q \le \widetilde M_{ij} \le 4 M_{ij}$ w.h.p.
		
		\STATE $\widetilde{U}\gets \emptyset$
		\FOR{$i\in[n/m]$}
		\FOR{$j\in[n/m]$}
		
		\STATE $\widetilde U_{ij}\gets \emptyset$
		\STATE $\bar{M}_{ij}\gets 0$
		\ENDFOR
		\ENDFOR
		\STATE $G\gets$ GraphConstruction$(x_1,\dots,x_{n/m},y_1,\dots,y_{n/m})$
		
		\STATE we have $G = (V,E)$ with $V = L \cup U \cup R$
		
		\FOR{$(\sigma,\ell,r)\in U$}
		
		\STATE $p_{\ell,r}\gets\min\{1,2^{\ell+r+3}/q\}$
		
		\STATE $q_{\ell,r}\gets$ 1 with probability $p_{\ell,r}$ and $0$ otherwise
		
		\IF{$q_{\ell,r}=1$}
		
		\STATE $\widetilde{U}\gets \widetilde{U}\cup \{(\sigma,\ell,r)\}$
		
		\ENDIF

		\ENDFOR		
		
		\FOR{$(\sigma,\ell,r)\in \widetilde{U}$}
		\STATE $L_{(\sigma,\ell,r)}=\{i \in L \mid (i,(\sigma,\ell,r))\in E\}$
		
		\STATE $R_{(\sigma,\ell,r)}=\{j \in R \mid (j,(\sigma,\ell,r))\in E\}$
		\ENDFOR	
		
		\FOR{$(\sigma,\ell,r)\in \widetilde{U}$}
		\FOR{$i\in L_{(\sigma,\ell,r)}$}
		\FOR{$j\in R_{(\sigma,\ell,r)}$}
		
		\STATE $\widetilde U_{ij}\gets \widetilde U_{ij}\cup \{(\sigma,\ell,r) \}$
		\ENDFOR
		\ENDFOR
		\ENDFOR	
		
		\FOR{$i\in[n/m]$}
		\FOR{$j\in[n/m]$}
		
		\FOR{$(\sigma,\ell,r)\in \widetilde U_{ij}$}
		\STATE $\overline{M}_{ij}\gets \overline{M}_{ij}+2^{\ell+r}/p_{\ell,r}$
		
		\ENDFOR
		\ENDFOR
		\ENDFOR
		
		\STATE $\widetilde M_{ij}\gets $ repeat steps $1-23$ for $c\log n$ times and find the median of all $\overline{M}_{ij}$ for each $i,j$
		\RETURN $\widetilde M_{ij}$ for each $i,j$
		
		\caption{{[Lemma \ref{lem:countmatchingpairs}] Approximating Matching Pairs} $(x_1,\dots,x_{n/m},y_1,\dots,y_{n/m},q)$}
		\label{alg:tool4}
	\end{algorithmic}
\end{algorithm}

\subsubsection{Single Symbol Approximation Algorithm}

\begin{algorithm}[H]
	\begin{algorithmic}[1]
		\REQUIRE Strings $x,y$ of length $n$ 
		
		\ENSURE Computes a common subsequence $\widetilde{L}$ of $x,y$ satisfying $\frac{M}{2n}\le \widetilde{L}\le L$

		\STATE $\Sigma\gets$ Compute the set of symbols appearing in $x,y$

		\STATE Compute the frequencies $\fr_\sigma(x)$ for all $\sigma \in \Sigma$ by one linear scan over $x$
		
		\STATE Compute the frequencies $\fr_\sigma(y)$ for all $\sigma \in \Sigma$ by one linear scan over $y$

		\STATE $\widetilde{L}\gets \max_{\sigma\in \Sigma}\min \{\fr_\sigma(x),\fr_\sigma(y)\}$

		\RETURN $\widetilde{L}$
		\caption{{[Lemma \ref{lem:apxsinglesymbol}] Single Symbol Approximation} $(x,y)$}
		\label{alg:tool5}
	\end{algorithmic}
\end{algorithm}


\begin{algorithm}[H]
	\begin{algorithmic}[1]
		\REQUIRE Strings $x_1,\dots,x_{n/m},y_1,\dots,y_{n/m}$ of length $m$ and a parameter $q>0$
		
		\ENSURE For each $i,j\in[n/m]$ computes $\widetilde L_{ij}$ satisfying $\frac{M_{ij} - q}{16 m}\le \widetilde L_{ij} \le L_{ij}$ w.h.p.
		
		\STATE Repeat steps $1-19$ of Algorithm \ref{alg:tool4}
		
		\FOR{$i\in[n/m]$}
		\FOR{$j\in[n/m]$}

		\STATE $\bar{L}_{ij} \gets \max_{(\sigma,\ell,r)\in \widetilde U_{ij}}2^{\min\{\ell,r\}} $
		
		\ENDFOR
		\ENDFOR
		
			\STATE $\widetilde L_{ij}\gets $ repeat steps $1-4$ for $c\log n$ times and find the median of all $\bar{L}_{ij}$ for each $i,j$

		\RETURN $\widetilde L_{ij}$ for each $i,j$
		
		\caption{{[Lemma \ref{lem:groupapxsinglesymbol}] Single Symbol Approximation Var} $(x_1,\dots,x_{n/m},y_1,\dots,y_{n/m},q)$}
		\label{alg:tool6}
	\end{algorithmic}
\end{algorithm}

\subsection{Main Algorithm}

In this section we describe the pseudocode for the main algorithm and it's four subparts.

\begin{algorithm}[H]
	\begin{algorithmic}[1]
		\REQUIRE Strings $x,y$ of length $n$; parameter $T\in [n,n^2]$ 
		
		\ENSURE Computes a common subsequence $\widetilde{L}$ of $x,y$ satisfying $\widetilde{L}=\widetilde{\Omega}(LT^{0.4}/n^{0.8})$ w.h.p.

		\STATE $\widetilde{L}_1\gets$ result of Algorithm \ref{alg:algo1} on $(x,y,T)$
		
			\STATE $\widetilde{L}_2\gets$ result of Algorithm \ref{alg:algo2} on $(x,y,T)$
		
			\STATE $\widetilde{L}_3\gets$ result of Algorithm \ref{alg:param} on $(x,y,T)$
		
		
		\STATE $\widetilde{L}\gets \max\{\widetilde{L}_1,\widetilde{L}_2,\widetilde{L}_3\}$

		\RETURN $\widetilde{L}$

		\caption{{[Theorem \ref{thm:maintechnical}] ApproxLCS} $(x,y,T)$}
		\label{alg:main}
	\end{algorithmic}
\end{algorithm}

\subsubsection{Algorithm $1$: Small $L$}

\setcounter{algorithm}{0}

\begin{algorithm}[H]
	\begin{algorithmic}[1]
		\REQUIRE Strings $x,y$ of length $n$; parameter $T\in [n,n^2]$ 
		
		\ENSURE Computes a common subsequence $\widetilde{L}$ of $x,y$ satisfying $\widetilde{L} \ge \min\{L,\sqrt{LT/n}\}$ w.h.p.
		
		\STATE $M\gets$ result of Algorithm \ref{alg:matching} on $(x,y)$
		
		\STATE $\widetilde{L}_1\gets$ result of Algorithm \ref{alg:tool5} on $(x,y)$
		
		\STATE $\widetilde{L}_2\gets$ $Query(x,y,\max\{1,\frac{M}{2T}\})$ of Algorithm \ref{alg:tool2}
		
		\STATE $\widetilde{L}\gets \max\{\widetilde{L}_1,\widetilde{L}_2\}$

		\RETURN $\widetilde{L}$

		\caption{{[Theorem \ref{thm:alg1}] Small $L$} $(x,y,T)$}
		\label{alg:algo1}
	\end{algorithmic}
\end{algorithm}

\subsubsection{Parameter Guessing}

\setcounter{algorithm}{13}
\begin{algorithm}[H]
	\begin{algorithmic}[1]
		\REQUIRE Strings $x,y$ of length $n$; parameter $T\in [n,n^2]$ 
		
		\ENSURE Computes a common subsequence $\widetilde{L}$ of $x,y$.
		
		\STATE $\widetilde{L}\gets 0$
		
		\FOR{$0\le i\le \log n$}
		
		\FOR{$0\le j\le 2\log n$}
		
		\FOR{$0\le k\le 2\log n$}
		
		\STATE $\hat{L}\gets 2^i$
		\STATE $\hat{\lambda}\gets 2^j$
		\STATE $\hat{\mu}\gets 2^k$
		
		\STATE $L_1\gets$ result of Algorithm \ref{alg:algo3} on $(x,y,T,\hat{L},\hat{\lambda},\hat{\mu})$, abort after running for time $\tOh(T)$
		
		\STATE $L_2\gets$ result of Algorithm \ref{alg:algo4} on $(x,y,T,\hat{L},\hat{\lambda},\hat{\mu})$, abort after running for time $\tOh(T)$
		
		\STATE $\widetilde{L}\gets \max\{\widetilde{L},L_1,L_2\}$
		
		\ENDFOR
		
		\ENDFOR
		
		\ENDFOR
		
		\RETURN $\widetilde{L}$
		
		\caption{{Parameter Guessing} $(x,y,T)$}
		\label{alg:param}
	\end{algorithmic}
\end{algorithm}

\subsubsection{Algorithm $2$: Large $L$, Large $\mu$}
\setcounter{algorithm}{1}

\begin{algorithm}[H]
	\begin{algorithmic}[1]
		\REQUIRE Strings $x,y$ of length $n$; parameter $T\in [n,n^2]$ 
		
		\ENSURE Computes a common subsequence $\widetilde{L}$ of $x,y$ satisfying $\widetilde{L} \ge \frac{\mu\sqrt{T}}{32n}$ w.h.p.
		
		\STATE $\mathcal{X}\gets \{x_1,\dots,x_{n/m} \}$ where $x_i=x[(i-1)m+1,im]\}$
		
		\STATE $\mathcal{Y}\gets \{y_1,\dots,y_{n/m} \}$ where $y_i=y[(i-1)m+1,im]\}$
		
		\STATE $M\gets$ result of Algorithm \ref{alg:matching} on $(x,y)$
		
		\STATE $q\gets \frac{M}{4T}$
		
		\STATE $\{\widetilde L_{ij} \}_{ i,j\in[n/m]}\gets$ result of Algorithm \ref{alg:tool6} on $(\mathcal{X},\mathcal{Y},q)$
		
		\FOR{$0\le i\le n/m$}
		
		\STATE $D[i,0]=0$
		\STATE $D[0,i]=0$
		
		\ENDFOR

		\FOR{$i\in [n/m]$}
		\FOR{$j\in [n/m]$}
		
		\STATE $D[i][j]\gets \max\{\widetilde L_{ij}+D[i-1,j-1],D[i-1,j],D[i,j-1]\}$
		
		\ENDFOR
		\ENDFOR

		\RETURN $D[n/m,n/m]$

		\caption{{[Theorem \ref{thm:alg2}] Large $L$, Large $\mu$} $(x,y,T)$}
		\label{alg:algo2}
	\end{algorithmic}
\end{algorithm}

\newpage
\subsubsection{Algorithm $3$: Large $L$, Small $\mu$, Large $\lambda$}

\begin{algorithm}
	\begin{algorithmic}[1]
		\REQUIRE Strings $x,y$ of length $n$; parameter $T\in [n,n^2]$ and parameters $\hat{L},\hat{\lambda},\hat{\mu}$ 
		
		\ENSURE Computes a common subsequence $\widetilde{L}$ of $x,y$ with the guarantees of Theorem \ref{thm:alg3}.

		\STATE $\mathcal{X}\gets \{x_1,\dots,x_{n/m}\}$ where $x_i=x[(i-1)m+1,im]\}$
		
		\STATE $\mathcal{Y}\gets \{y_1,\dots,y_{n/m}\}$ where $y_i=y[(i-1)m+1,im]\}$
		
		\STATE $\bar{L}\gets 0$
		
		\STATE $\alpha\gets \max\{1,\hat{\lambda}/(4T)\}$
		
		\FOR{$\log \alpha \le i\le \log m$}
	
		\STATE $g\gets 2^i$
		\STATE $\mathcal{S}\gets \emptyset$
		
		\STATE $p\gets \frac{cgTm^2\log^2 n}{\hat{\lambda} n^2}$

		\FOR{$i\in[n/m]$}
		\FOR{$j\in[n/m]$}
		
		\STATE $q\gets$ $1$ with probability $p$ and $0$ otherwise
		
		\IF{$q=1$}
		
	   \STATE $f\gets Query(x_i,y_j,g)$ of Algorithm \ref{alg:tool3}
	   
	   \IF{$f=1$}
	   
	   \STATE $\mathcal{S}\gets \mathcal{S}\cup \{(i,j)\}$
	 
        \ENDIF
        
        \ENDIF
        
        \ENDFOR
        
        \ENDFOR
        
        \IF{$\mathcal{S} \ne \emptyset$}
        
        \STATE $(i_0,j_0)\gets$ select a block uniformly at random from $\mathcal{S}$
        
        \STATE $d\gets i_0-j_0$
        
        \STATE $\mathcal{D}\gets \{(i,j)\in[n/m]^2 \mid i-j=d\}$
        \STATE $\beta \gets \max\{1,\hat{\mu}/T\}$
        \FOR{$(i,j)\in \mathcal{D}$}
        
        \STATE $\widetilde{L}_{ij}\gets$ result of Algorithm \ref{alg:tool2} on $(x_i,y_j,\beta)$
        
        \ENDFOR
	
		\STATE $\widetilde{L}(g)\gets \sum_{(i,j)\in \mathcal{D}}\widetilde{L}_{ij}$
		
		\STATE $\bar{L}\gets \max\{\widetilde{L},\widetilde{L}(g)\}$
        \ENDIF
		
		\ENDFOR

		\STATE $\widetilde L\gets $ repeat steps $5-24$ for $b\log n$ times and compute the maximum of all $\bar{L}$ 
		
		\RETURN $\widetilde{L}$

		\caption{{[Theorem \ref{thm:alg3}] Large $L$, Small $\mu$, Large $\lambda$} $(x,y,T,\hat{L},\hat{\lambda},\hat{\mu})$}
		\label{alg:algo3}
	\end{algorithmic}
\end{algorithm}

\newpage
\subsubsection{Algorithm $4$: Large$L$, Small $\mu$, Small $\lambda$}

\begin{algorithm}
	\begin{algorithmic}[1]
		\REQUIRE Strings $x,y$ of length $n$; parameter $T\in [n,n^2]$ and parameters $\hat{L},\hat{\lambda},\hat{\mu}$ 
		
		\ENSURE Computes a common subsequence $\widetilde{L}$ of $x,y$ with the guarantees of Theorem \ref{thm:alg4}.
		
		\STATE $\mathcal{X}\gets \{x_1,\dots,x_{n/m}\}$ where $x_i=x[(i-1)m+1,im]\}$
		
		\STATE $\mathcal{Y}\gets \{y_1,\dots,y_{n/m}\}$ where $y_i=y[(i-1)m+1,im]\}$
		
		\STATE $M\gets$ result of Algorithm \ref{alg:matching} on $(x,y)$
		
		\STATE $q\gets \frac{M}{T}$
		
		\STATE $\{\widetilde M_{ij}\}_{i,j\in[n/m]}\gets$ result of Algorithm \ref{alg:tool4} on $(\mathcal{X},\mathcal{Y},q)$
		
		\FOR{$i\in[n/m]$}
		\FOR{$j\in[n/m]$}
		
		\STATE $\widetilde L_{ij}\gets 0$
		
		\ENDFOR
		\ENDFOR
		
		\FOR{$j\in[n/m]$}
		\STATE Run the preprocessing of basic decision algorithm (Lemma \ref{lem:basicdec}) on $y_j$ 
		\ENDFOR
		
		\STATE $p\gets \min\Big\{ \frac{\hat L}n, \frac{\hat L T}{\hat \lambda n}, \frac{\hat L^2 T}{\hat \lambda \hat \mu n} \Big\}$
		
		\FOR{$i\in[n/m]$}
		\FOR{$j\in[n/m]$}
		
		\STATE $q\gets$ $1$ with probability $p$ and $0$ otherwise
		
		\IF{$q=1$}

		\IF{$\widetilde M_{ij}\le 64\hat{\mu}n/(\hat{L}\sqrt{T})$}
		
		\STATE $f\gets$ $Query(x_i,y_j,\frac{\hat{L}}{4\sqrt{T}})$ of Algorithm \ref{alg:tool3}
		
		\IF{$f=1$}
		
		\STATE $\widetilde L_{ij}\gets \frac{\hat{L}}{4\sqrt{T}}$
		
		\ENDIF

		\ENDIF

		\ENDIF
		
		\ENDFOR
		\ENDFOR

		\FOR{$0\le i\le n/m$}
		
		\STATE $D[i,0]=0$
		\STATE $D[0,i]=0$
		
		\ENDFOR

		\FOR{$i\in [n/m]$}
		\FOR{$j\in [n/m]$}
		
		\STATE $D[i][j]\gets \max\{\widetilde L_{ij}+D[i-1,j-1],D[i-1,j],D[i,j-1]\}$
		
		\ENDFOR
		\ENDFOR

		\RETURN $D[n/m,n/m]$

		\caption{{[Theorem \ref{thm:alg4}] Large $L$, Small $\mu$, Small $\lambda$} $(x,y,T,\hat{L},\hat{\lambda},\hat{\mu})$}
		\label{alg:algo4}
	\end{algorithmic}
\end{algorithm}